\documentclass[12pt]{article}
\usepackage{epsf}

\usepackage{epsfig}
\usepackage{amssymb}
\usepackage{amsmath}
\usepackage{times}
\usepackage{./mhequ}

\newenvironment{proof}[1][Proof]{\noindent\textit{#1.} }{\ \rule{0.5em}{0.5em}\par}

\newcommand{\ed}{{\rm e}}
\newcommand{\boundC}{\alpha}

\newcommand{\be}{\begin{equs}}
\newcommand{\ee}{\end{equs}}
\newcommand{\bea}{\begin{equs}}
\newcommand{\eea}{\end{equs}}
\newcommand{\beas}{\begin{equs*}}
\newcommand{\eeas}{\end{equs*}}

\newtheorem{theorem}{Theorem}[section]
\newtheorem{lemma}[theorem]{Lemma}

\newtheorem{proposition}[theorem]{Proposition}
\newtheorem{corollary}[theorem]{Corollary}
\newtheorem{remark}[theorem]{Remark}

\newtheorem{hypothesis}[theorem]{Hypothesis}

\newcommand{\real}{{\bf R}}

\def\One{\mathbb{I}}

\newcommand{\myl}[1]{
\hspace{-1mm}
\left(\vbox to #1pt{}\right.
\hspace{-1mm}
}

\newcommand{\my}[2]{
\left#1\vbox to #2pt{}\right.
}

\newcommand{\myr}[1]{
\hspace{-1mm}
\left.\vbox to #1pt{}\right)
\hspace{-1mm}
}

\def\backgr{\;\raisebox{-36mm}{\epsfysize=100mm\epsfbox{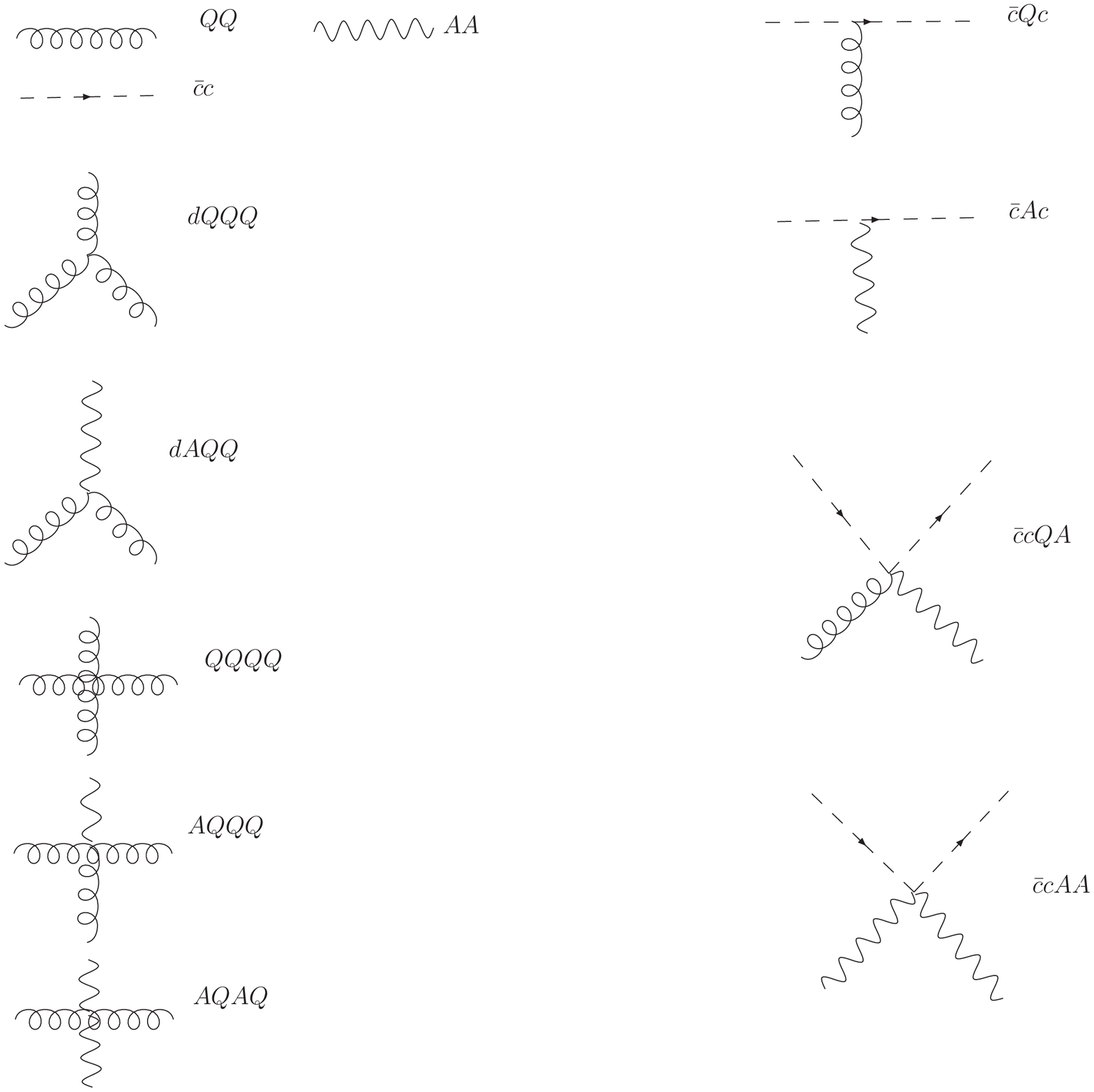}}\;}

\begin{document}
\title{The QCD $\beta$-function from global solutions to Dyson-Schwinger equations}

\author{Guillaume van Baalen\footnote{Department of Mathematics and
    Statistics, Boston University, 111 Cummington Street, Boston MA,
    02215, USA}, Dirk Kreimer\footnote{CNRS-IHES 91440 Bures sur
    Yvette, France; and Center for Mathematical Physics Boston
    University, 111 Cummington Street, Boston MA,
    02215, USA}, David Uminsky\footnotemark[1], and Karen
  Yeats\footnotemark[1]}
\maketitle
\begin{abstract}
We study quantum chromo\-dynamics from the viewpoint of
untruncated Dy\-son--Schwin\-ger equations turned to an ordinary
differential equation for the gluon anomalous dimension. This
non\-linear equation is parameterized by a function $P(x)$ which
is unknown beyond perturbation theory. Still, very mild
assumptions on $P(x)$ lead to stringent restrictions for possible
solutions to Dy\-son--Schwin\-ger equations.

We establish that the theory must have asymptotic freedom beyond
perturbation theory and also investigate the low energy regime
and the possibility for a mass gap in the asymptotically free
theory.
\end{abstract}
\section*{Acknowledgments}  D.K.\ and
K.Y.\ were supported by NSF grant DMS-0603781. D.U.\ and G.v.B.\
are supported by NSF grant DMS-0405724 and thank C.E.\ Wayne for
discussions. 
D.K.\ thanks Ivan Todorov
for discussions.
\section{Introduction}
We study non-perturbative aspects of quantum chromodynamics
(QCD). We do so by investigating Dyson--Schwinger equations.
Instead of solving a truncated version of these a priori very
intricate equations \cite{Marciano:1978ee}, we use recent insight
into the mathematical structure of quantum field theory to gain
insight into the possible structure of solutions. This approach
has been successfully applied to quantum electrodynamics in
\cite{QED} and is here extended to QCD.
\subsection{The method}
We follow the methods employed in our work on QED \cite{QED},
adopted to the study of QCD in the background field approach as
developed by Abbott \cite{Abbott1,Abbott2}. Let us first
reconsider the situation for quantum electrodynamics. There,
thanks to the Ward identity, it suffices to consider the
anomalous dimension $\gamma_1(x)$ of the photon which is
essentially the $\beta$-function,
$\beta(x)=x\gamma_1(x)$.

This anomalous dimension is obtained from the photon's
self-energy, a two-point function which is determined
non-perturbatively by the knowledge of a single Lorentz scalar
function
\begin{equs}
G(x,L)=1-\sum_{k=1}^\infty \gamma_k(x)L^k~,
\end{equs}
with $x$ the fine structure constant and $L=\ln
\left(-q^2/\mu^2\right)$, where the inverse photon-propagator is
$(q^2g_{\mu\nu}-q_\mu q_\nu)G(x,L)$.

Combining this with the combinatorial Dyson-Schwinger equations
and using an expansion into suitable integral kernels which
parametrize the corresponding integral equation, the
Dyson-Schwinger equation combine with the renormalization group
equation to give

i) a recursion for the $\gamma_k$:
\begin{equs}
\gamma_k(x)=-\frac{1}{k}\gamma_1(x)(1-x\partial_x)\gamma_{k-1}(x),
k\geq 2~,
\label{eqn:recu}
\end{equs}

ii) a differential equation for $\gamma_1(x)$:
\begin{equs}
\gamma_1(x)(1-x\partial_x)\gamma_1(x)+\gamma_1(x)-P(x)=0~.
\label{eqn:DSE}
\end{equs}

Here, i) comes from the renormalization group, ii) from the DSE,
and $P(x)$ is a suitably constructed series over residues.
The appearance of the operator $(1\pm x\partial_x)$ is typical for a
gauge theory.

We will now address a non-abelian gauge theory, resulting in a
similar set-up (in particular, i) and ii) remain form-invariant)
once we learned to make effective use of quantum gauge invariance
to reduce again to a single ODE. Obviously, however, $P(x)$ will
be a different function, in particular, it will change sign. We
study the consequences of this fact with minimal knowledge on the
behavior of $P(x)$. Still, as in QED, we will see that we can
learn quite a bit regarding the non-perturbative sector of QCD.

To proceed, we turn to the background field method.
\subsection{Background field method}
We first have to consider the set of vertices and propagators in
the background field gauge.
They define a set $\mathcal{R}$ given as follows:
\begin{equs}
\backgr.
\end{equs}
Here, names attached to the vertices and propagators are short
hand memos of the corresponding monomials in the Lagrangian.

The field content is $Q$ for the internal quantized gauge field,
$A$ for an external background gauge field, $\bar{c},c$
for the anti-ghost and ghost field. Coupling of Fermionic matter
will not change the ensuing discussion in any way and is omitted for
convenience. See Abbott \cite{Abbott2} for details.

With the set $\mathcal{R}$ comes an accompanying set of 1PI
Feynman graphs naturally labeled by elements in this set
according to their type and number of external legs.

We consider in particular Green functions for such graphs and
adopt the results of \cite{Kre05}, see also \cite{KvS}, which
read as expected $\forall r\in \mathcal{R}$
\begin{equs}
X^r=\One\pm \sum_{k\geq 1} [g^2]^k 
\sum_{|\gamma|=k}\frac{1}{\mathrm{sym}(\gamma)}B_+^{\gamma;r}
\left(\frac{\prod_{v\in\gamma^{[0]}}X^v}{\prod_{e\in\gamma^{[1]}}
\sqrt{X^e}}\right)~,
\end{equs}
with 
\begin{equs}
B_+^\gamma(h)=\sum_{\Gamma\in<\Gamma>}
\frac{{\textbf{bij}(\gamma,h,\Gamma)}}{|h|_\vee}\frac{1}
{\textrm{maxf}(\Gamma)}\frac{1}{(\gamma|h)}\Gamma~,\label{def}
\end{equs}
where maxf$(\Gamma)$ is the number of maximal forests of
$\Gamma$, $|h|_\vee$ is the number of distinct graphs obtainable
by permuting edges of $h$, $\textbf{bij}(\gamma,h,\Gamma)$ is the
number of bijections of external edges of $h$ with an insertion
place in $\gamma$ such that the result is $\Gamma$, and finally
$(\gamma|h)$ is the number of insertion places for $h$ in
$\gamma$ \cite{Kre05}. $\sum_{\Gamma\in <\Gamma>}$ indicates a
sum over the linear span $<\Gamma>$ of generators of $H$.

Next, we divide by the ideal $I$ which implements the
Slavnov--Taylor identities which here is generated order in order
in $g^2$ by
\begin{equs}
X^{\bar{c}Ac} & =
X^{\bar{c}Qc}=X^{\bar{c}c},\;X^{dQQQ}=X^{dAQQ}=X^{QQ}=X^{AA}~,\\
X^{AAQQ}& = X^{AQAQ}=X^{AQQQ}=X^{QQQQ}=X^{dQQQ}~.
\end{equs}
On $H/I$ we then get two independent Green functions which need
renormalization
\begin{equs}
X^{AA}, \;X^{\bar{c}c}~,
\end{equs}
corresponding to a mere two-element set
\begin{equs}
\mathcal{R}_{H/I}=\{AA,\bar{c}c\}~.
\end{equs}
Also, we then find an combinatorial invariant charge uniquely
defined as
\begin{equs}
C=\One/\sqrt{X^{AA}}=\One/\sqrt{X^{QQ}}~,
\end{equs}
so that, as in QED, the $\beta$-function is just half the
negative anomalous dimension of the gauge field.

Note that the addition of massless fermions would just add in
$H/I$ an element $\bar{\psi}\psi$ for the fermion self-energy
but would not change the ideal or the invariant charge.

The system of combinatorial Dyson Schwinger equations is then
\begin{equs}
X^{AA} & =  \One-\sum_k [g^2]^k B_+^{k,AA}\left(X^{AA}\left[C\right]^{-2k}\right)\\
X^{\bar{c}c} & = \One-\sum_k [g^2]^k
B_+^{k,\bar{c}c}\left(X^{\bar{c}c}\left[C\right]^{-2k}\right)~.
\end{equs}
Note that this determines $X^{AA}$ in terms of itself, while
$X^{\bar{c}c}$ is a function of itself and $X^{QQ}$. We write
$X^r=\One-\sum_{k\geq 1} [g^2]^k c_k^r$, with $c_k^r\in H/I$ the
generators of a sub Hopf algebra \cite{Kre05,BergKr06} given by
all graphs which contribute to an amplitude $r$ at a chosen order
$k$. Note that the resolution of a Green functions into images of
Hochschild one-cocycles $B_+^r$ is the mathematical equivalent of
a resolution of all overlapping divergences into non-overlapping
integral kernels. That this is possible in a non-abelian theory
was realized early by Baker and Lee \cite{Baker:1976vz}.

See \cite{Kre05,KvS,Sui07,Sui08} for explicit examples how
Hochschild cohomology and Hopf algebras relate.

There hence is a quotient Hopf algebra $H_{AA}$ spanned by
generators $c^{AA}_k$ of $X^{AA}$. Similarly, going temporarily
to the quotient Hopf algebra $H_{\bar{c}c}$ defined by
$X^{AA}=\One$, $c_k^{AA}=0$, we find that this is a cocommutative
Hopf algebra (which is obvious from setting $C=\One$) and hence
we find a factorization of groups into an abelian subgroup
$\mathrm{Spec}(H_{\bar{c}c})$ and a normal subgroup
$\mathrm{Spec}(H_{AA})$,
\begin{equs}
\mathrm{Spec}(H/I)=\mathrm{Spec}(H_{AA})\rtimes
\mathrm{Spec}(H_{\bar{c}c})~,
\end{equs}
corresponding to a short exact sequence which splits
\begin{equs}
\One\to \mathrm{Spec}(H_{AA})\to \mathrm{Spec}(H/I)\to
\mathrm{Spec}(H_{\bar{c}c})\to\One~.
\end{equs}
It is this factorization which allows us to compute the
$\beta$-function of QCD by an ODE for a single equation below.
Note that the situation is similar to QED: there, the Ward
identity allows for a similar semi-direct product structure
between photon amplitudes and Fermionic matter. Gauge invariance
then allows to eliminate all short-distance singularities in the
abelian subgroup thanks to the work of Baker, Johnson and Willey,
and one is left with the photon propagation as the only source of
renormalization.

Here, we can compute the $\beta$-function from $G^{AA}(x,L)$, but
would have to consider the full coupled system to determine
$G^{\bar{c}c}(x,L)$ (and $G^{\bar{\psi}\psi}(x,L)$), which we do
not attempt here.

Furthermore, note that the simplification at $k=1$
\begin{equs}
B_+^{1,AA}\left(X^{AA}\left[C\right]^{-2}\right)=B_+^{1,AA}(\One)
~.
\end{equs}
This is typically for gauge theories and emphasizes that we are
in a single equation situation with $s=1$ \cite{QED}. Terms
$B_+^r(\One)$ always deliver a pure residue from their
short-distance singularities, and these terms are intimately
connected to fermion determinants. We will not pursue this
connection any further here.

The background field method is then suited to our approach as it
allows us to compute the QCD beta function from a single ordinary
differential equation.

Indeed, a change of basis of primitives allows to reduce the
application of Feynman rules to the study of one-variable Mellin
transforms for the integral kernels for the above primitives
$B_+^{k,r}(\One)$, and from there we can strictly follow the
techniques of \cite{KY06,KY,Y} to get to an single ordinary
differential equation:
\begin{equs}
\gamma_1(x)+\gamma_1(x)^2-P(x)-x\gamma_1(x)\gamma_1^\prime(x)=0~.
\end{equs}
Here, $P(x)$ is a suitable series over primitives. As always, we
renormalize using a momentum scheme with subtractions at
$q^2=\mu^2$. That scheme is uniquely suited \cite {CelmGon} to
our gauge-invariant non-perturbative approach.

\subsection{Qualitative properties of QCD}
From perturbative computations, asymptotic freedom is firmly
established. We will establish it beyond perturbation theory
below. Perturbatively, this is mainly a self-consistency
statement: assuming that the QCD coupling constant is small, we
approximate the theory by its loop expansion to a few orders. The
resulting polynomial approximation to the beta function supports
the claim of asymptotic freedom in perturbation theory:
$\beta(x)<0$, $0<x<1$, hence at large momentum transfer
$\lim_{-Q^2\to\infty} \alpha_s(L)\to 0$, $L=\ln( -Q^2/\mu^2)$.
The coupling indeed becomes small in that limit.

As usual, perturbation theory agrees well with observations:
asymptotic freedom is a well-established experimental fact.

Much more intricate is the study
\begin{equs}
-Q^2\to 0_+~.
\end{equs}
This is beyond the reach of perturbation theory. Nevertheless,
different approaches point out that in that limit, the gluon
propagator might turn to a constant, confirming an old suggestion
of Cornwall that the free gluon develops in the interacting
theory a momentum-dependent mass which vanishes at high energies,
but turns to a non-vanishing constant in the limit $-Q^2\to 0_+$.
We study this behavior from our viewpoint in section \ref{gap}.

Results to this effect were already obtained by\\
i) Lattice computations \cite{Aguilar:2008xm};

ii) numerical study of Dyson Schwinger equations truncated in a
gauge invariant way \cite{Boucaud:2008ky};

iii) in the Gribov-Zwanziger formalism \cite{john}.

Below, we reconsider the problem from a study of the possible
structure of solutions of Dyson Schwinger equations. We want to
establish asymptotic freedom beyond perturbation theory, and want
to discuss to what extent a solution which exhibits asymptotic
freedom can also exhibit a mass gap.

Again, as in the case of QED, we find the most interesting
solution to be a separatrix. In the case of QCD, that separatrix
is the only solution which has asymptotic freedom.

\section{Results}

\label{sec:QCD}
In QCD, the Dyson-Schwinger equation for $\gamma_1$ is
\begin{equs}
\label{DSeqnQCDfi}
\frac{{\rm d}\gamma_1(x)}{{\rm d}x} = 
f(\gamma_1(x),x)\equiv
\frac{\gamma_1(x)
+\gamma_1(x)^2-P(x)}{x\gamma_1(x)}~.
\end{equs}
We will assume that the primitive skeleton function satisfies the
following assumptions:
\begin{itemize}
\item[H1:] $P$ is a twice differentiable function on ${\bf R}^{+}$, with 
$P(0)=0$, $P'(0)<0$ and $P''(0)<0$.
\item[H2:] There exist $x^{\star}$ such that $P(x)>-\frac{1}{4}$
and $P''(x)\leq0$ (i.e. $P$ is concave) on
$[0,x^{\star}]$.
\item[H3:] The function $P(x)$ satisfies $P(x)<0$ for all $x>0$.
\end{itemize}
As in \cite{QED}, we avoid the singularities of
(\ref{DSeqnQCDfi}) at $\gamma_1=0$ and $x=0$ by specifying with
an initial condition at $x^{\star}$, namely,
\begin{equs}
\label{DSeqnQCD}
\frac{{\rm d}\gamma_1(x)}{{\rm d}x} = 
f(\gamma_1(x),x)\equiv
\frac{\gamma_1(x)
+\gamma_1(x)^2-P(x)}{x\gamma_1(x)}~,~~~
\gamma_1(x^{\star})=\gamma_0~.
\end{equs}
This ensures that solutions of (\ref{DSeqnQCD}) exist at least
locally around $x=x^{\star}$. Though we will mainly look for
solutions or (\ref{DSeqnQCDfi}) with $\gamma_1(x^{\star})<0$ and
$x\geq0$, we will occasionally comment on the
$\gamma_1(x^{\star})>0$ case.

In the QED case, we proved in \cite{QED} (see also Section
\ref{sec:QED} of the present paper) the existence of a unique
value $\gamma_1^{\star}(x_0)$ (the separatrix) separating
solutions that exist globally for all $x\geq x_0$ from those that
can only be continued up to a finite $x_{\rm max}>x_0$. As shown
in Section \ref{sec:QED}, all solutions in QED can be continued
as $x\to0$, differing there `only' by a flat behavior
$\sim\ed^{-\frac{1}{x}}$. In QCD the situation is reversed: we
will prove that all solutions starting at some appropriate
$x^{\star}$ can be continued as $x\to\infty$, but that there is a
unique value $\gamma_1^{\star}(x^{\star})$ that separates
solutions that cannot be continued as $x\to0$ from those that
can, which either satisfy $\gamma_1(0)=-1$ if
$\gamma_1(x^{\star})<\gamma_1^{\star}(x^{\star})$ or
$\gamma_1(0)=0$ if
$\gamma_1(x^{\star})=\gamma_1^{\star}(x^{\star})$. We call the
solution $\gamma_1^{\star}(x)$ that satisfies
$\gamma_1^{\star}(0)=0$ the {\em asymptotically free} solution.

We will also use more speculative hypotheses on $P(x)$:
\begin{itemize}
\item[S1:] There exists $p>0$ such that $P(x)=-c x^p+{\it o}(x^p)$
as $x\to\infty$.
\item[S2:] There exist $x_c\geq -\frac{1}{P'(0)}$ such that 
$P''(x)\leq0$ for all $x\in[0,x_c]$
\item[S3:] There exist a (finite) interval $[x_{c},x_{d}]$ with
$x_{c}>x^{\star}$ such that
\begin{equs}
-\int_{x_{c}}^{x_{d}}\frac{1+4P(z)}{2z}{\rm d}z\geq1~.
\end{equs}
\item[S4:] There exist finite $x_l$ and $x_r$ such that 
$P(x_l)=P(x_r)=-\frac{1}{4}$ and $P(x)>-\frac{1}{4}$ for all 
$0\leq x<x_l$ and $x>x_r$. The function $P(x)$ does {\em not}
satisfy S1, S2 and S3, but rather
$\lim_{x\to\infty}(P(x)-P_{\infty})=\lim_{x\to\infty}xP'(x)=0$
for some $P_{\infty}>-\frac{1}{4}$.
\end{itemize}
Let us briefly comment on our logic here. The H1-H3 hypotheses
are a bare minimum for the results we will present below. Within
perturbation theory, we have
\begin{equs}
P(x)=\gamma_1(x)+{\cal O}(x^3)=-\beta_1 x- \beta_2 x^2+{\cal O}(x^3)
\label{eqn:pertu}
\end{equs}
as $x\to0$, where $-\beta_1$ and $-\beta_2$ are the $1$ and $2$-loop
coefficients of the $\beta$ function, namely $\beta_1=9$ and
$\beta_2=64$ for $n_f=6$. As such, the hypotheses H1 and H2 are
reasonable. While it also follows from (\ref{eqn:pertu}) that
$P(x)<0$ at least for small values of $x$, extending this to all
values of $x$ is somewhat more speculative.

As we will show below, if $P(x)<-\frac{1}{4}$ on a `sufficiently
large interval', for instance if either S1, S2 or S3 hold, all
solutions of (\ref{DSeqnQCD}) satisfy $\gamma_1(x)=-1$ for some
$x\geq x^{\star}$. It then follows from {\rm H3} that they grow
linearly as $x\to\infty$ if
\begin{equs}
{\cal D}(P)=
-\int_{x^{\star}}^{\infty}\frac{P(z)}{z^3}{\rm d}z<\infty~,
\end{equs}
and faster than linearly if ${\cal D}(P)=\infty$. Incidentally,
we showed in \cite{QED} (see also Section \ref{sec:QED} of the
present paper), that the finiteness/infiniteness of ${\cal D}(P)$
was intimately linked with the existence/non-existence as
$x\to\infty$ of solutions of the analogous of (\ref{DSeqnQCD})
for QED. It is striking to see that the {\em same criterion}
distinguishes between different type of behavior in QCD as
well\footnote{though of course the primitive skeleton functions
$P$ are different in both cases}.

Note that an anomalous dimension growing at least linearly as
$x\to\infty$ leads to a Landau pole for the running
coupling, and hence a serious obstacle to studying the infrared
behavior of the Gluon propagator. Despite that, we will show in
Section \ref{sec:Dirk} that this pole (if present) can be removed
using an unsubtracted dispersion relation, see e.g.
\cite{Shirkov}, and the Gluon propagator can still be studied in
the infrared limit.

In contrast, if $\gamma_1(x)$ is finite as $x\to\infty$, we avoid
the Landau pole, and can study the Gluon propagator without using
dispersion relations. Such constant asymptotics for $\gamma_1$
can only happen if $P(x)$ tends to a constant as $x\to\infty$.
This motivates (part of) the hypothesis S4. Under that
hypothesis, we will show that there is only one solution
that satisfies
\begin{equs}
\lim_{x\to\infty}\gamma_1(x)=-\frac{1+\sqrt{1+4P_{\infty}}}{2}
\equiv\gamma_{\infty}~.
\end{equs}
If $P_{\infty}=0$, we call that solution the {\em confinement
solution} $\gamma_1^{c}(x)$, and if $P_{\infty}>0$, we call it a
{\em strong confinement solution}. On physical grounds, the
asymptotically free and (strong) confinement solutions need to be
the same. Unfortunately, for generic $P(x)$ satisfying H1, H2 and
S4, these two solutions are different. We did not succeed in
finding a sufficient condition on $P(x)$ that guarantees both
solutions are the same. Despite that, a necessary condition is
certainly that $P(x)$ makes at least one (small) excursion below
$-\frac{1}{4}$, while avoiding the S1-S3 conditions, see also
figure \ref{fragile} below or Section \ref{sec:confsol}.

By standard folklore and heuristics \cite{Acharya}, a nowhere
vanishing $\beta$-function, $\beta(x)<0,\;\forall x>0$, which we will
indeed establish below under assumption H1 above, implies a mass
gap in QCD, and hence a confinement scenario following old ideas
of Cornwall \cite{Cornwall}. We will come back to that in Section
\ref{gap} of this paper.

\begin{remark}
\label{rem:zero}
If $P_{\infty}>0$ in hypothesis S4, $P(x)$ has a further zero for
a finite $x_1>0$, $P(x_1)=0$. In such a case, using the running
coupling formulation of (\ref{DSeqnQCDfi}), one sees that some
solutions spiral around the zero of $P(x)$, themselves having
infinitely many zeroes. The asymptotically free solution
$\gamma_1^{\star}(x)$ may or may not spiral around $x=x_1$,
depending on details of $P$. Should it be captured, we get a
solution $\gamma_1(x)$ which has a UV fix-point at zero and an
infrared fixpoint at $x_1$. This is the Banks-Zaks scenario
\cite{BZ}.
\end{remark}

In the remainder of this section, we are going to state our main
results. The proofs and technical details are postponed to
Section \ref{sec:technicalities} of this paper.

Our first main result gives a complete characterization of the
behavior of solutions of (\ref{DSeqnQCD}) for $x<x^{\star}$. In
particular, it establishes the uniqueness of the solution
exhibiting asymptotic freedom.
\begin{theorem}\label{thm:completechar}
Under the hypotheses {\rm H1} and {\em H2}, there is a unique
value $\gamma_1^{\star}(x^{\star})<0$ such that the corresponding
solution $\gamma_1^{\star}(x)$ of (\ref{DSeqnQCD}) exists for all
$x\in[0,x^{\star}]$ and satisfies
${\displaystyle\lim_{x\to0}}\gamma_1^{\star}(x)=0$. Additionally,
that solution satisfies
\begin{equs}
\frac{\sqrt{1+4P(x)}-1}{2}
\leq \gamma_1^{\star}(x)\leq P'(0)x
~,
\label{eqn:sandwich}
\end{equs}
If $\gamma_1(x^{\star})<\gamma_1^{\star}(x^{\star})$, then the
corresponding solution satisfies
${\displaystyle\lim_{x\to0}}\gamma_1(x)=-1$. If
$\gamma_1^{\star}(x^{\star})<\gamma_1(x^{\star})<0$, then there
exist $x_{\rm min}>0$ such that the corresponding solution
satisfies $\gamma_1(x_{\rm min})=0$.
\end{theorem}
Note that since $\frac{\sqrt{1+4P(x)}-1}{2}=P'(0)x+{\cal O}(x^2)$
as $x\to0$, (\ref{eqn:sandwich}) shows that
$\gamma_1^{\star}(x)=P'(0)x+{\cal O}(x^2)$ as $x\to0$.

As we will prove in Proposition \ref{prop:negglobal}) of Section
\ref{sec:technicalities} below, all solutions of Theorem
\ref{thm:completechar} can be continued as $x\to\infty$ by only
adding the H3 assumptions. In Proposition
\ref{prop:doublevalued}, we will show that solutions that satisfy
$\gamma_1(x_{\rm min})=0$ for some $0<x_{\rm min}<x^{\star}$ can
be continued in the first quadrant (becoming double-valued) by
reverting to the so-called `running coupling' formulation of
(\ref{DSeqnQCD}) (see also \cite{QED}). In particular these
solutions will satisfy $\gamma_1(x_0)>0$ for some $x_0>x_{\rm
min}$.

Our second main result concerns the asymptotic behavior as
$x\to\infty$ of solutions that enter the first quadrant, or
attain the value $-1$ somewhere. This last condition can be
verified under additional assumptions on $P$ such as S1, S2 or
S3.
\begin{proposition}
\label{prop:minusone}
Assume $P(x)$ satisfies H1-H3 and that one of the two
following statements holds:
\begin{enumerate}
\item $-1<\gamma_1(x^{\star})<0$ and $P(x)$ satisfies S1 or S3,
\item $\gamma_1(x^{\star})\leq\gamma_1^{\star}(x^{\star})$ and
$P(x)$ satisfies S2.
\end{enumerate}
Then there exists $x_0>x^{\star}$ such that the corresponding
solution $\gamma_1(x)$ satisfies $\gamma_1(x_0)=-1$.
\end{proposition}

To be able to state our asymptotic result as $x\to\infty$, we
need first to introduce the {\em slope function} ${\rm
S}_{P}(x_0,x)$. This function is given by
\begin{equs}
{\rm S}_{P}(x_0,x)=
\left(
\frac{\gamma_1(x_0)^2}{x_0^2}
+2
\int_{x_0}^{x}
\frac{-P(z)}{z^3}{\rm d}z
\right)^{\frac{1}{2}}~,
\end{equs}
Note that, if ${\cal D}(P)<\infty$, the slope function ${\rm
S}_{P}(x_0,x)$ goes to a finite value as $x\to\infty$ for any
$x_0$.

We can now completely describe the asymptotic behavior as
$x\to\infty$ of solutions of (\ref{DSeqnQCD}):
\begin{theorem}
\label{thm:asxtoinfty}
Assume $P(x)$ satisfies H1-H3. If there exists $x_0>0$ such
that either $\gamma_1(x_0)>0$ or $\gamma_1(x_0)\leq-1$ then
\begin{equs}[3]
 x~S_{P}(x_0,x)
\leq
\gamma_1(x)
&\leq 
\phantom{-}
x~
\myl{18}
S_{P}(x_0,x)+
\frac{1}{x_0}\myr{18}-1 &
\mbox{~~~if~~~}&\gamma_1(x_0)>0~,\\
-x~S_{P}(x_0,x)
\leq
\gamma_1(x)
&\leq 
-x~
\myl{18}
S_{P}(x_0,x)-
\frac{1}{x_0}\myr{18}-1 &
\mbox{~~~if~~~}&\gamma_1(x_0)\leq-1
~.
\end{equs}
Furthermore, if ${\cal D}(P)<\infty$ and $\gamma_1(x_0)\leq-1$ or 
$\gamma_1(x_0)>0$, there exists $s>0$ such that
\begin{equs}
\lim_{x\to\infty}\frac{\gamma_1(x)}{x}=
\left\{
\begin{array}{rl}
-s<0 & \mbox{if}~~\gamma_1(x_0)\leq -1\\[2mm]
 s>0 & \mbox{if}~~\gamma_1(x_0)> 0
\end{array}\right.
~.
\end{equs}
If $\gamma_1(x_0)\leq-1$, the convergence towards the limit is
given by
\begin{equs}
\my{|}{14}\frac{\gamma_1(x)}{x}+s\my{|}{14}
\leq C\int_{x}^{\infty}
\frac{-P(z)}{z^3}{\rm d}z~.
\label{eqn:linearatinfone}
\end{equs}
If ${\cal D}(P)<\infty$ and $\gamma_1(x_0)>0$, then
(\ref{eqn:linearatinfone}) also hold, with $-s$ replaced by $s$.
\end{theorem}

We want to stress here that the slope value $s$ depends on the
actual solution under consideration. Also, for solutions that
eventually enter the first quadrant, the value of the slope $s$
has no reason to be the same along the two branches of the
solution (the one in the first quadrant, and the one in the
fourth). Also, note that the result depends only on the
assumption $\gamma_1(x_0)=-1$.

We now state the existence and uniqueness of {\em confinement
solutions} under hypothesis S4.
\begin{theorem}
\label{thm:confinement}
Assume $P(x)$ satisfies H1, H2 and S4. Then there exist a unique
solution $\gamma_1^{c}(x)$ of (\ref{DSeqnQCDfi}) such that
\begin{equs}
\lim_{x\to\infty}\gamma_1^{c}(x)=
\lim_{x\to\infty}-\frac{1+\sqrt{1+4P(x)}}{2}\equiv
\gamma_{\infty}~.
\end{equs}
Furthermore, $\gamma_1^c(x)$ satisfies the usual trichotomy as
$x$ decreases: either $\gamma_1^c(0)=-1$, or $\gamma_1^c(0)=0$,
or $\gamma_1^c(x)$ cannot be continued for $x<x_{\rm min}$ for
some $x_{\rm min}>0$ where $\gamma_1^c(x_{\rm min})=0$. Finally,
if there exists $x_{\rm max}$ such that $'P(x)>0$ for all
$x>x_{\rm max}$ (hence $P(x)$ is strictly increasing towards
$P_{\infty}$), then
\begin{equs}
-\frac{\sqrt{1+4P_{\infty}}+1}{2}
\leq
\gamma_1^{c}(x)
\leq -\frac{\sqrt{1+4P(x)}+1}{2}
\end{equs}
for all $x>x_{\rm max}$.
\end{theorem}
As already noted above, we cannot show that
$\gamma_1^{\star}(x)=\gamma_1^c(x)$ without additional hypotheses
on $P(x)$. We can however note that the two types of solutions
are compatible: the confinement solution satisfies the integral
equation
\begin{equs}
\gamma_1^c(x)=-1
+\int_{1}^{\infty}\frac{P(xt)}{t^2\gamma_1^c(xt)}{\rm d}t
~,
\label{eqn:integralforconf}
\end{equs}
whose r.h.s.\ converges to $0$ as $x\to0$ if
$P(x)/\gamma_1^{c}(x)\to1$ as $x\to0$. Characterizing the set of
functions $P(x)$ for which the confinement solution and the
asymptotically free one are the same is a difficult problem, see
e.g. figure \ref{fragile} for an example with an `artificial'
$P(x)$. A necessary condition is that $P(x)$ makes an excursion
below $-\frac{1}{4}$ on (at least) one interval so that the
nullclines, i.e.\ the location in $(x,\gamma_1)$-plane where
$\gamma_1'(x)=0$,
\begin{equs}
\gamma_{\pm}(x)=
\frac{\pm\sqrt{1+4P(x)}-1}{2}
\end{equs}
show a gap as in figure \ref{fragile}. However if the gap is `too
wide', we have $\gamma_1^{\star}(x)=\gamma_{\infty}$ for some
finite $x$, and $\gamma_1^{c}(x_{\rm min})=0$ for some $x_{\rm
min}>0$, whereas if the gap is `too small', $\gamma_1^{\star}(x)$
intersects the nullcline at some $x>x_r$, and hence cannot reach
$\gamma_{\infty}$ as $x\to\infty$, while $\gamma_1^{c}(x)$
intersects the nullcline at some $x<x_l$, and hence cannot reach
$0$ as $x\to0$.
\begin{figure}
\begin{center}
\unitlength1mm
\begin{picture}(80,110)(0,0)
\put(0,0){
\epsfig{file=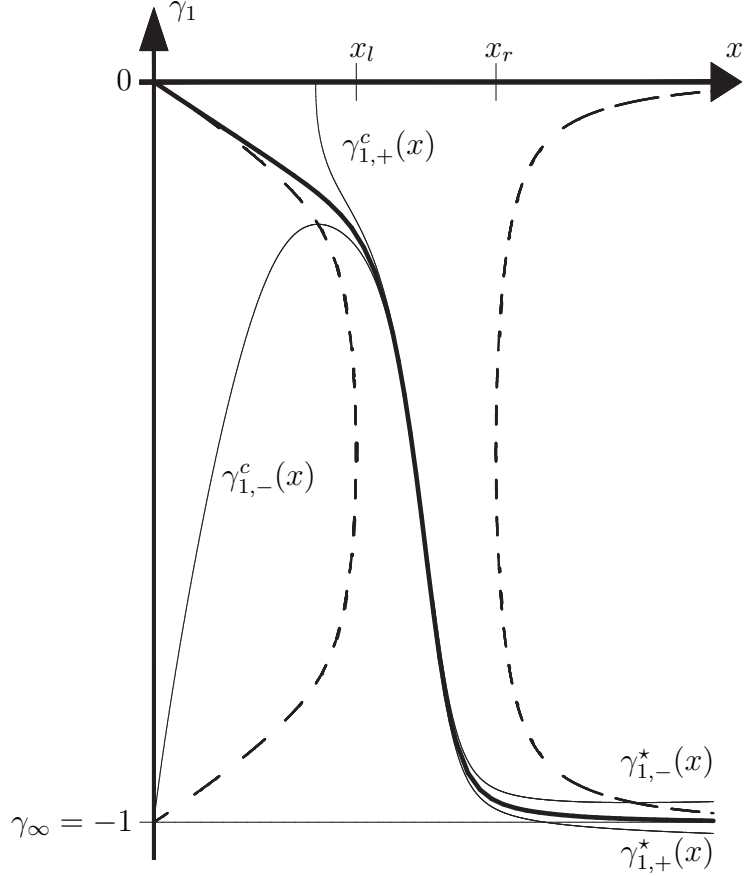,width=8cm}
}
\put(-16,4){$\gamma_{\infty}=-1$}
\put(-2,102.5){$0$}
\put(5,112.5){$\gamma_1$}
\put(79,107){$x$}
\put(29,107){$x_l$}
\put(47,107){$x_r$}
\put(12,50){$\gamma_{1,-}^{c}(x)$}
\put(28,94){$\gamma_{1,+}^{c}(x)$}
\put(65,0){$\gamma_{1,+}^{\star}(x)$}
\put(65,12){$\gamma_{1,-}^{\star}(x)$}
\end{picture}
\end{center}
\caption{A `fragile' solution: the solid bold curve is the
solution combining asymptotic freedom and confinement for some
artificial $P(x)$ with $P_{\infty}=0$. The dashed curves are the
nullclines $\gamma_1+\gamma_1^2-P(x)=0$. The solid curves
$\gamma_{1,\pm}^{c}$ and $\gamma_{1,\pm}^{\star}$ correspond to
confinement and asymptotically free solutions with $P(x)$
perturbed so that the gap $|x_r-x_l|$ is slightly larger (`$+$'
subscripts) or smaller (`$-$' subscripts). In both cases, the
asymptotically free solution and the confinement solution do not
match.}
\label{fragile}
\end{figure}

We conclude this section by explaining the terminology {\em
asymptotically free}, {\em confinement} and {\em strong
confinement} for our solutions. These come from the study of the
inverse Gluon propagator, 
\begin{equs}
P^{-1}(x,Q^2)=Q^2G(x,L)~~~\mbox{with}~~~
L=\ln(Q^2/\mu^2)~.
\end{equs}
In Section (\ref{sec:uncorrected}) below, we solve the RGE
equation expressing the scale invariance of $G(x,L)$.
In particular, in Theorem \ref{thm:transport}, we show that if
$\gamma_1(x)\to\gamma_{\infty}$ as $x\to\infty$, then
\begin{equs}
G(x,L)=\frac{X(L,x)}{x}~,
\end{equs}
where $X(t,x)$ is the running coupling, i.e.\ the solution of
\begin{equs}
\frac{{\rm d}X(t,x)}{{\rm
d}t}=X(t,x)\gamma_1(X(t,x))~~~\mbox{with}~~~X(t=0,x)=x~.
\end{equs}
The possible large $|t|$ behavior of $X(t,x)$ are given by
\begin{equs}
X(t,x)\simeq\left\{
\begin{array}{lll}
x_{\infty}\ed^{\gamma_{\infty}t} &~~\mbox{as}~~t\to-\infty &
\mbox{if}~~~{\displaystyle\lim_{x\to\infty}}\gamma_1(x)=\gamma_{\infty}~,\\[4mm]
-\frac{1}{P'(0)t} &~~\mbox{as}~~t\to\infty &
\mbox{if}~~~{\displaystyle\lim_{x\to0}}\frac{\gamma_1(x)}{x}=P'(0)~,
\\[4mm]
x_0\ed^{-t} &~~\mbox{as}~~t\to\infty &
\mbox{if}~~~{\displaystyle\lim_{x\to0}}\gamma_1(x)=-1~,
\end{array}
\right.
\end{equs}
where $x_{\infty}$ and $x_0$ are some positive functions of $x$.
Thus, in all cases where $\gamma_1(x)$ can be continued to $x=0$,
$G(x,L)\to0$ in the ultraviolet regime $L\to\infty$, and moreover
\begin{equs}
P^{-1}(x,Q^2)\to 
\left\{
\begin{array}{lll}
0&~~~\mbox{as}~~~L\to\infty
&\mbox{if}~~~{\displaystyle\lim_{x\to0}}\frac{\gamma_1(x)}{x}=P'(
0)~,\\[4mm]
+x_0&~~~\mbox{as}~~~L\to\infty
&\mbox{if}~~~{\displaystyle\lim_{x\to0}}\gamma_1(x)=-1~,
\end{array}
\right.
\end{equs}
hence the terminology {\em asymptotic freedom} for the solution
that satisfies $\gamma_1(x)=xP'(0)+{\cal O}(x^2)$ as $x\to0$. In
the infrared regime, however,
\begin{equs}
G(x,L)\simeq\frac{x_{\infty}}{x}
\ed^{\gamma_{\infty}L}
=
\frac{x_{\infty}}{x}
\left(
\frac{Q^2}{\mu^2}
\right)^{\gamma_{\infty}}
~~~\mbox{as}~~~L\to-\infty~,
\end{equs}
and hence we find
\begin{equs}
\lim_{L\to-\infty}P^{-1}(x,Q^2)
=
\left\{
\begin{array}{cl}
\mu^2~\frac{x_{\infty}}{x} &
~~~\mbox{if}~~~\displaystyle\lim_{x\to\infty}\gamma_1(x)=-1\\[3mm]
\infty&
~~~\mbox{if}~~~\displaystyle\lim_{x\to\infty}\gamma_1(x)<-1
\end{array}
~,
\right.
\label{eqn:confiture}
\end{equs}
a finite mass gap if $\gamma_{\infty}=-1$.

\section{The Gluon propagator, confinement and mass gaps in
QCD}\label{gap}

\subsection{Corrections from dispersion
relations}\label{sec:Dirk}

In QED, the coupling is weak at low energy or momentum transfer.
In our previous work \cite{QED}, we could hence define boundary
conditions at low energy, and studied the behavior of
$\sum_k\gamma_k L^k$ for $L>0$. In particular, for
$\gamma_1=\gamma_1(\bar{x}(L))$, a continuation to $L<0$ was
never needed by the choice of our renormalization conditions. On
the other hand, for large $L\gg 0$ we could establish a
separatrix, which possibly avoids a Landau pole at any finite
positive $L$. We conjectured that this might be the solution
chosen by Nature, and further detailed analysis of its properties
awaits more analysis of the function $P(x)$. Should it turn out
that $P(x)$ is such that the separatrix will not avoid a Landau
pole ({\em i.e.}\ turns to infinity at finite $L$), we will have to turn
to dispersion relations to understand the non-perturbative
corrections coming with such a pole, as recognized by Shirkov and
collaborators early on \cite{ShirkovEarly}.

For QCD, we again fix a small coupling, but this time large
momentum transfer, $L\gg 0$ for our boundary conditions. We are
now interested in a continuation to $L\ll0$, in particular we are
interested in $L\to-\infty$. Under very mild assumptions on
$P(x)$, and certainly by any experience from perturbative
approximations of the theory, we expect the anomalous dimension
$\gamma_1$ to go below the value $-1$ at some finite coupling
$x$, and hence $\bar{x}(L)$ to turn to infinity at some finite
negative $L$. Shifting that $L$ to zero essentially defines the
scale $\Lambda_{\mathrm{QCD}}$, and we are interested for that
shifted $L_\Lambda$ to study the regime $L_\Lambda<0$, in
particular $L_\Lambda\to -\infty$. To consider such a limit based
from an approach formulated for $L_\Lambda>0$, we will use
dispersion relations. Our approach is motivated again by Shirkov
and collaborators work \cite{Shirkov}.

On general grounds, we know that $\bar{x}(L_\Lambda)$ and
$G(x,L_\Lambda)$ can be treated by an unsubtracted dispersion
relation \cite{OZ}:
\begin{equs}
f_{\mathrm{disp}}(Q^2)=\int_0^\infty \frac{\Im
(f(\sigma))}{\sigma+Q^2-i\eta}{\rm d}\sigma~.
\end{equs}
The inverse propagator needs a subtracted dispersion relation,
which leads back to an unsubtracted dispersion relation for $G$
\cite{OZ}.

The Dyson--Schwinger equations themselves are supposed to hold
for the whole theory regardless of the sign of $L_\Lambda$.
Similar, the renormalization group equations for the running
coupling are supposed to hold. Our derivation which turned the
Dyson--Schwinger equations into a ODE was valid for
$L>-L_\Lambda$, hence remain valid, after shifting, for
$L_\Lambda>0$.

Continuing to $L_\Lambda<0$ will generate non-perturbative
corrections to $\gamma_1(x)$, and hence $\beta(x)$, determined
from the requirement that equations of motion, renormalization
group flow and analyticity properties of field theory are what
they are supposed to be.

Any perturbative approximation is in accordance with these
properties of field theory only up to the order considered. When
we study solutions of Dyson--Schwinger equations, we demand {\em
accord} with these properties as a guide to find the necessary
non-perturbative corrections in the region $L_\Lambda<0$.

We hence will start by first applying a dispersion relation to
analyze $\bar{x}(L_\Lambda)$. That leads to a corrected
$\gamma_{1,\mathrm{disp}}$ due to the fact that at $L_\Lambda=0$
we find that $\bar{x}=\infty$. Combining this with the assumption
that the uncorrected $\gamma_1(x)$ is driven below -1 at finite
$x$ allows an easy estimate of the corrected
$\gamma_{1,\mathrm{disp}}$ and also allows for consistency with
the direct analysis of $G$ by dispersion methods.

Let us now start with a study of $\bar{x}(L_\Lambda)$. We start
our considerations by boundary conditions such that $G(x,L)=1$ at
some very high momentum transfer $\mu^2\gg 0$, $L=\ln(Q^2/\mu^2)$,
which determines a suitably small $x$ in agreement with say deep
inelastic scattering experiments \cite{exp}.

We have
\begin{equs}
\frac{{\rm d}\bar{x}(x,L)}{{\rm d}L}=\bar{x}\gamma_1(\bar{x})\Leftrightarrow
\int_x^{\bar{x}(x,L)}\frac{1}{u\gamma_1(u)}du=L~,
\end{equs}
and assume that the integral
\begin{equs}
\int_x^\infty \frac{1}{u\gamma_1(u)}{\rm d}u<\infty~.
\end{equs}
In particular, we assume that $\gamma_1(x)<-1$ for some positive
finite $x$, in accordance with our previous discussions and
experimental evidence. We thus define 
\begin{equs}
\Lambda_{\mathrm{QCD}}^2=\mu^2+e^{\left\{ \int_x^\infty
\frac{1}{u\gamma_1(u)}{\rm d}u \right\}}~.
\end{equs}
Setting $L_\Lambda=\ln(Q^2/\Lambda_{\mathrm{QCD}}^2)$ to absorb the dependence
on $x,\mu^2$, we get
\begin{equs}
-\int_{\bar{x}(L_\Lambda)}^\infty
\frac{1}{u\gamma_1(u)}{\rm d}u=L_\Lambda~.
\label{eqn:defLl}
\end{equs}
This equation defines $L_\Lambda(\bar{x})$ as well as the inverse
function $\bar{x}(L_\Lambda)$. Since $\bar{x}(0)=+\infty$, we
cannot trust our solution for $Q^2<\Lambda^2$. To
remedy this, we use our previous result that for large $x$,
$\gamma_1(x)\to -sx$, under our present assumptions. 

Following the conventions of Shirkov \cite{Shirkov}, we define a
running coupling in accordance with the expected analytic
behavior of field theory using a dispersion relation:
\begin{equs}
\bar{x}_{\mathrm{disp}}(Q^2)=\frac{1}{\pi}\int_0^\infty 
\frac{\Im(\bar{x}(\ln(\sigma/\Lambda^2)))}{\sigma+Q^2-i\eta}
{\rm d}\sigma~.
\label{eqn:dispersedef}
\end{equs}
The pole at $-Q^2$ gives us back the uncorrected
$\bar{x}(L_\Lambda)$. But by assumption, there is a further pole
in the complex $\sigma$-plane, located at $L_\Lambda(\infty)=0$.
To study the contribution from that pole, we first note that as
$\bar{x}\to\infty$, we have
\begin{equs}
L_{\Lambda}=
-\int_{\bar{x}(L_\Lambda)}^\infty
\frac{1}{u\gamma_1(u)}{\rm d}u\simeq
\int_{\bar{x}(L_{\Lambda})}^\infty \frac{1}{su^2}~{\rm
d}u
=\frac{1}{s\bar{x}(L_{\Lambda})}~,
\end{equs}
and hence $\bar{x}(L_{\Lambda})\simeq\frac{1}{sL_{\Lambda}}$ near
$L_{\Lambda}=0$. Feeding this relation into
(\ref{eqn:dispersedef}) gives (see \cite{Shirkov})
\begin{equs}
\bar{x}_{\mathrm{disp}}(L_{\Lambda})=\bar{x}(L_{\Lambda})
+\frac{1}{s(1-\ed^{L_{\Lambda}})}
~.
\end{equs}
If we were to identify $s$ with the one-loop coefficient of the
$\beta$-function, this would reproduce Shirkov's analysis for
one-loop QCD, see \cite{Shirkov}, where Shirkov also notes that 
$s$ seem not to vary much at low loop orders. Note that the
correction to $\bar{x}(L_{\Lambda})$ goes to the finite value
$1/s$ in the infrared limit $L_{\Lambda}\to-\infty$.

Now, $\gamma_1$ also obeys an unsubtracted dispersion relation.
As $\gamma_1(x)$ is finite for all finite $x$ and depends on $L$
only through $\bar{x}$, the dispersion integral will correct
$\gamma_1(\bar{x}(L))\simeq -s\bar{x}(L_{\Lambda})$ by 
\begin{equs}
\gamma_{1,\mathrm{disp}}(\bar{x}(L_\Lambda)) =
\gamma_1(\bar{x}(L_{\Lambda}))
-\frac{1}{1-\ed^{L_{\Lambda}}}~.
\end{equs}
Note now that the correction to $\gamma_1$ goes to $-1$ as
$L_{\Lambda}\to-\infty$. Using (\ref{eqn:recu}) and
$x\gamma_1(x)\partial_x=\partial_L$, we get 
\begin{equs}
\gamma_{k,{\rm disp}}(\bar{x}(L_{\Lambda}))
=\gamma_{k}(\bar{x}(L_{\Lambda}))+
\gamma_{k,{\rm corr}}(L_{\Lambda})~~~\mbox{with}~~~
\gamma_{k,{\rm
corr}}(L_{\Lambda})\to\frac{(-1)^k}{k!}~~\mbox{as}~~~~
L_{\Lambda}\to-\infty~.
\end{equs}
As such, the correction from the dispersion relation to the
function $G(x,L_{\Lambda})$ satisfies
\begin{equs}
G_{\rm corr}(x,L_{\Lambda})\to
\sum_{k=0}^{\infty}\frac{(-1)^k}{k!}L_{\Lambda}^k
=\ed^{-L_{\Lambda}}=
\frac{\Lambda^2}{Q^2}~~~\mbox{as}~~~L_{\Lambda}\to-\infty~,
\end{equs}
which gives an inverse propagator satisfying
\begin{equs}
\lim_{L_{\Lambda}\to-\infty}
P^{-1}(x,Q^2)=-\Lambda^2~.
\end{equs}
This gives a finite and renormalization group invariant
positivity-violating mass gap. This compares nicely with the
results of Gracey et.al.\ \cite{john}.

Note that we rely completely on the assumption that physical
quantities in massless field theory have neither poles nor branch
cuts off the negative (in our conventions) real axis and the
result that the solutions we get from our ODEs for the
uncorrected $\bar{x},\gamma_k,G(x,L_\Lambda)$ are in accordance
with these requirements but for the isolated pole at
$L_\Lambda=0$.

If a future analysis of $P(x)$ justifies the trust in field
theory expressed in this section remains to be seen. We are
content having identified a clean mechanism for the generation of
a mass gap, based on the consequences of the Hopf algebra
structure underlying local field theory, and assuming analyticity
properties in accordance with the underlying axiomatic structure
of local quantum fields.

Let us add a few comments which put our results in context. We  rely on the
assumptions on the function $P(x)$ outlined above. Note that the ODE for the anomalous dimension $\gamma_1(x)$ 
takes into account all the iterations of superficially divergent graphs into each other, hence all renormalon
ambiguities met in the context of resummation of a perturbative series are taken care of. We hope this will make $P(x)$, which
has an asymptotic expansion related to a weighted skeleton expansion, amenable to more constructive methods of analysis in the future,
in contrast to $\gamma_1(x)$ whose resummation as a perturbative series faces such ambiguities.

On the other hand, we emphasize that our analysis avoids any truncation of Dyson--Schwinger equations, and any assumptions made 
on the infrared behaviour and infrared powercounting of QCD amplitudes. It thus complements the approaches available so far in the literature,
see for example \cite{IRFischer}. While other methods based on Dyson--Schwinger equations attempt to solve them in some justifiable limit or truncation, we establish properties  of the solution of the full equations depending  on assumptions made for  $P(x)$.
\subsection{Solving The Renormalization Group Equation for
$G(x,L)$}
\label{sec:uncorrected}

Our goal in this section is to analyze the dressing function
$G(x,L)$ that modulates the free inverse propagator. On general
grounds, this function solves the Renormalization Group equation
\begin{equs}
\myl{12}
-\partial_L
+x\gamma_1(x)\partial_x
-s\gamma_1(x)
\myr{12}G(x,L)=0~~~\mbox{and}~~~G(x,0)=1~~~\mbox{for}~~~x\geq0~,
\label{eqn:rge}
\end{equs}
where $s=\pm1$ distinguishes between writing the propagator as
\begin{equs}
P(x,Q^2)=\frac{G(x,\ln(Q^2/\mu^2))^{s}}{Q^2}
\end{equs}
with $s=1$ or $s=-1$. Note that the difference of (\ref{eqn:rge})
from the usual RGE equation
\begin{equs}
\myl{12}
\mu\partial_{\mu}+\tilde{\beta}(g)\partial_g+\tilde{\gamma}(g)
\myr{12}G=0
\end{equs}
is merely a matter of convention on the definition of $\beta$,
$\gamma_1$ and $x$. In particular, (\ref{eqn:rge}) agrees with
two-loop computation of the Gluon propagator.

We will here make the following (minimal) hypothesis on
$\gamma_1(x)$:
\begin{hypothesis}
\label{hyp:hypopopo}
The function $\gamma_1(x)$ is a negative ${\cal
C}^1([0,\infty),(-\infty,0])$ function whose only possible zero
is at $x=0$, where $\gamma_1(x)=-d x^{q_0}+{\cal O}(x^{q_0+1})$
with $q_0\geq0$ and $d>0$.
\end{hypothesis}
Note that the results of Section \ref{sec:QCD} give $q_0=1$ and
$d=-P'(0)$ for the asymptotically free solution $\gamma_1(x)$, in
accordance with perturbation theory. For the solutions satisfying
$\gamma_1(0)=-1$ however, we get $q_0=0$ and $d=1$.

We first note that one can attempt to solve (\ref{eqn:rge}) by
writing
\begin{equs}
G(x,L)=1-s\sum_{k=1}^{\infty}\gamma_k(x)L^k~~~~\mbox{with}~~~~
\gamma_{k}(x)=
\frac{\gamma_1(x)}{k}
\myl{12}
x\partial_x-s
\myr{12}\gamma_{k-1}(x)~~~k\geq2~,
\label{eqn:seriessol}
\end{equs}
since one has (at least formally) that
\begin{equs}
\myl{12}
-\partial_L
+x\gamma_1(x)\partial_x
-s\gamma_1(x)
\myr{12}G(x,L)&=
-s\gamma_1(x)+s
\sum_{k=1}^{\infty}k\gamma_k(x)L^{k-1}
-\gamma_1(x)\myl{12}
x\partial_x
-s
\myr{12}
\gamma_k(x)L^k\\
&=
s\sum_{k=2}^{\infty}
\left(
\gamma_k(x)
-\frac{\gamma_1(x)}{k}
\myl{12}
x\partial_x-s
\myr{12}
\gamma_{k-1}(x)
\right)kL^k=0~.
\label{eqn:seriestwo}
\end{equs}
This approach naturally raises the (difficult) question of
convergence of the series in (\ref{eqn:seriessol}) and
(\ref{eqn:seriestwo}). We will use instead an alternative way of
solving (\ref{eqn:rge}) that avoids these convergence problems.
In particular, while the series (\ref{eqn:seriessol}) necessarily
converge on a symmetric interval of the form $(-L_0,L_0)$ for
some (possibly infinite) $L_0\geq0$, our approach will give a
solution to (\ref{eqn:rge}) that is defined on an interval of
the form $(-L_0,\infty)$ for the same $L_0$. As such, the
approach below shows that the limit $L\to\infty$ of $G(x,L)$
makes sense also if the series solution converges only on a
finite interval. 

Our method is based on the fact that (\ref{eqn:rge}) can be
transformed into a {\em linear transport equation} by
appropriately factorizing $G(x,L)$ into one part that cancels the
term involving no derivatives and one part that
solves a genuine transport equation of the form
\begin{equs}
(-\partial_L+x\gamma_1(x)\partial_x)H(x,L)=0~.
\label{eqn:truetransport}
\end{equs}
As such, it is important at first to consider the characteristics
curve of (\ref{eqn:truetransport}). For each fixed
$x>0$, we first define $X(t,x)$ as the solution of the running
coupling equation
\begin{equs}
\frac{{\rm d}X(t,x)}{{\rm
d}t}=X(t,x)\gamma_1(X(t,x))~~~\mbox{with}~~~X(t=0,x)=x~.
\label{eqn:RGEt}
\end{equs}
Since $\gamma_1(x)$ is assumed to be ${\cal C}^1$, solutions of this
equation exist at least locally around $t=0$. For further
reference, we denote by ${\cal D}(x)$ the maximal interval of
existence of the solution of (\ref{eqn:RGEt}) for a fixed $x$.

Then for each fixed $(x,L)\in{\bf R}^{+}\times{\bf R}$, we define
the characteristic curve ${\cal C}(x,L)$ as
\begin{equs}
{\cal C}(x,L)&=
\my{\{}{14}
(X(t,x),L-t)~~~\mbox{with}~~~t\in{\cal D}(x)
\my{\}}{14}~,\\
&=
\my{\{}{18}
\left(X,L-
\int_x^{X}\frac{{\rm d}z}{z\gamma_1(z)}
\right)~~~\mbox{with}~~~X\in{\bf R}^{+}
\my{\}}{18}~.
\label{eqn:intermsofX}
\end{equs}
Note that both above formulations of the characteristic curve
${\cal C}(x,L)$ are equivalent. The characteristics corresponding
to different values of $L$ are {\em vertical} translations of the
same curve in the $(x,L)$-plane. By hypothesis
\ref{hyp:hypopopo}, we get that the characteristics are all
asymptotically vertical as $X\to0$, and, as a function of $t$, we
have
\begin{equs}
X(t,x)\simeq
\left\{
\begin{array}{cll}
\ed^{-ct} & ~~~\mbox{as}~~~t\to\infty&~~~\mbox{if}~~~q_0=0\\
(ct)^{\frac{1}{-q_0}}& ~~~\mbox{as}~~~t\to\infty&~~~\mbox{if}~~~q_0>0
\end{array}
\right.~.
\label{eqn:ttoinf}
\end{equs}
The behavior of the characteristics as
$X\to\infty$ depends on the asymptotic behavior of $\gamma_1(x)$ as
$x\to\infty$. In all cases, ${\cal C}(x,L)$ approaches
$(\infty,L-L_{\infty}(x))$ as $X\to\infty$, where $L_{\infty}(x)$
is the {\em possibly infinite} quantity defined by
\begin{equs}
L_{\infty}(x)=\int_x^{\infty}\frac{{\rm d}z}{z\gamma_1(z)}~.
\end{equs}
This shows that the maximal interval of existence for solutions
of (\ref{eqn:RGEt}) is ${\cal D}(x)=(L_{\infty}(x),\infty)$. The
results of Section \ref{sec:QCD} show that $L_{\infty}(x)$ is
finite if there exists $x_0>0$ where $\gamma_1(x_0)=-1$, and
infinite otherwise. If $L_{\infty}(x)>-\infty$, the
characteristic ${\cal C}(x,L)$ intersects the line $L=0$ if and
only if $L>L_{\infty}(x)$. If $-1<\gamma_1(x)<0$ for all $x>0$,
then $L_{\infty}(x)=-\infty$ for all $x>0$, and {\em all}
characteristic curves cross the line $L=0$. The generic shape of
the characteristics curves is displayed in figure
\ref{alltogether}.
\begin{figure}
\begin{center}
\unitlength1mm
\begin{picture}(145,145)(0,0)
\put(-10,0){
\epsfig{file=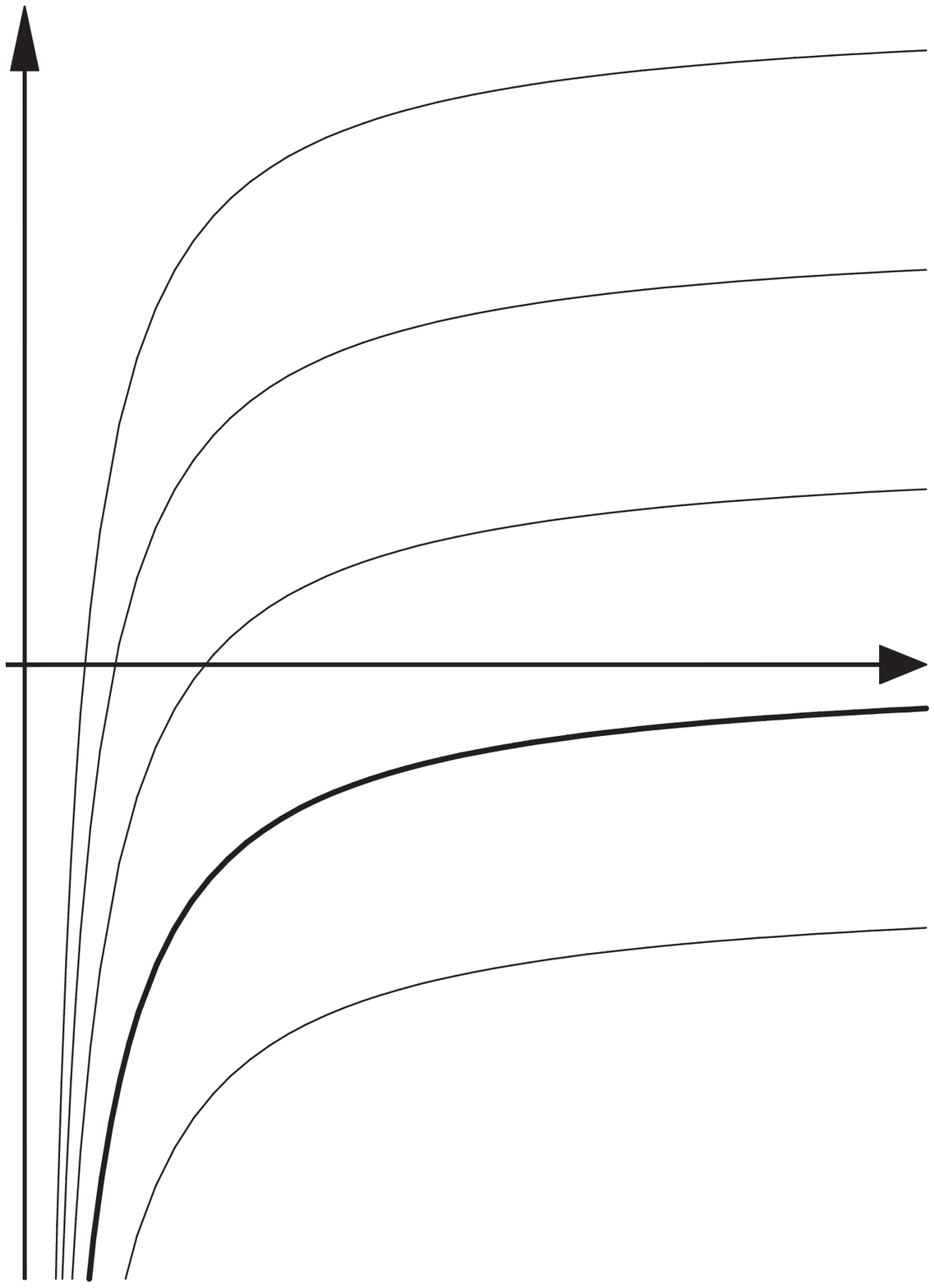,height=10cm}
\hspace{10mm}
\epsfig{file=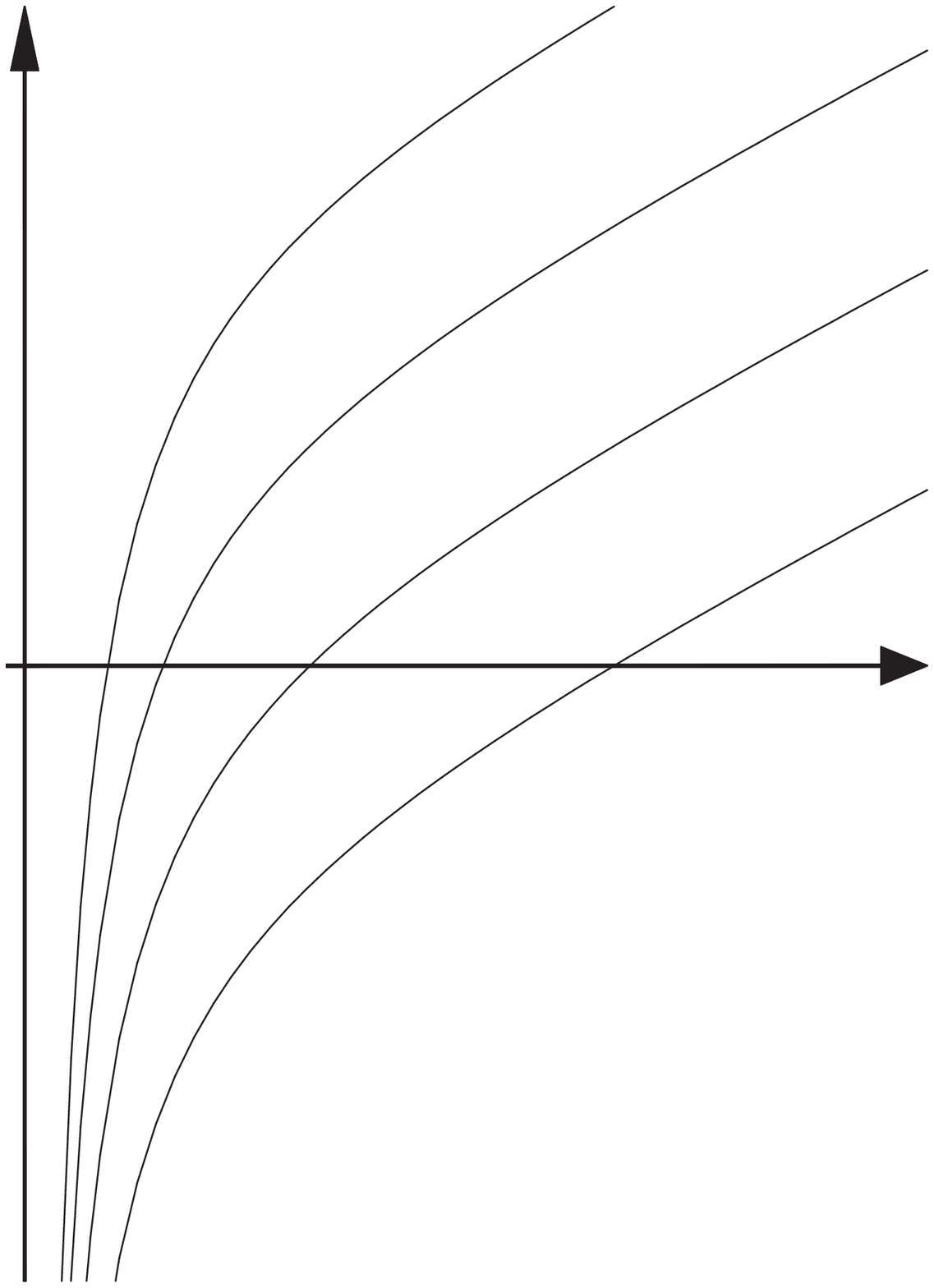,height=10cm}
}
\put(-5.5,97){$L$}
\put(60,51){$x$}
\put(79,97){$L$}
\put(145,51){$x$}
\put(25,36){$L=L_{\infty}(x)$}
\end{picture}
\end{center}
\caption{Generic shape of characteristics curves in
$(x,L)$-plane. The left panel shows the case where
$L_{\infty}(x)$ is finite, the right panel when it is infinite.
In the left panel, the characteristics cross the line $L=0$ if
and only if they are above the (bold) curve $L=L_{\infty}(x)$,
and all characteristics are asymptotically horizontal as
$x\to\infty$. In the right panel, all characteristics cross the
line $L=0$, as they link $(0,-\infty)$ to $(\infty,\infty)$.}
\label{alltogether}
\end{figure}

We can now give the solution of (\ref{eqn:rge}) in accordance (for $s=-1$, as expected) with Bogoliubov-Shirkov (\cite{BogShir}, App.IX, eq.(27)):
\begin{theorem}
\label{thm:transport}
Assume $\gamma_1(x)$ satisfies Hypothesis \ref{hyp:hypopopo}.
Then the solution of (\ref{eqn:rge}) is given by
\begin{equs}
G(x,L)&=\left(
\frac{x}{X(L,x)}
\right)^s
\label{eqn:simplephys}
\end{equs}
for all $(x,L)$ such that $L_{\infty}(x)<L<\infty$ and $x>0$.
\end{theorem}

We want to stress here that the relation (\ref{eqn:simplephys})
hold only for pair of values of $(x,L)$ satisfying
$L_{\infty}(x)<L<\infty$. For pairs $(x,L)$ with $L\leq
L_{\infty}(x)$, $G(x,L)$ can only be determined from
$G(x,0)=1$ if one specifies $\gamma_1(x)$ for $x<0$ as well.

Assuming $X(t,x)$ to be analytic for all $t\in{\cal D}(x)$, it is
a straightforward computation (see Section (\ref{sec:alternate}))
to show that (\ref{eqn:seriessol}) is the Taylor series expansion
of (\ref{eqn:simplephys}) at $L=0$. As such, the series solution
(\ref{eqn:seriessol}) is expected {\em not} to converge for
$|L|>L_{\infty}(x)$. However, the function provided by
(\ref{eqn:simplephys}) is defined for
unbounded positive $L$ and do solve (\ref{eqn:rge}) for all such
$L$. If $L_{\infty}(x)>-\infty$, Theorem \ref{thm:transport} thus
gives $G(x,L)$ for values of $(x,L)$ for which the series
formulation fails.

\begin{proof}[proof of Theorem \ref{thm:transport}]
We first set
\begin{equs}
G(x,L)=x^sH(x,L)
\label{eqn:splitG}
\end{equs}
for all $x>0$. Substitution into
(\ref{eqn:rge}) gives
\begin{equs}
\myl{12}
-\partial_L
+x\gamma_1(x)\partial_x
\myr{12}H(x,L)=0~~~\mbox{and}~~~H(x,0)=
x^{-s}
~~~\mbox{if}~~~x>0~.
\label{eqn:eqnforH}
\end{equs}
We then note that $H(x,L)$ is constant along the characteristic
curve ${\cal C}(x,L)$, for we have
\begin{equs}
\frac{{\rm d}}{{\rm d}t}
\myl{12}
H(X(t,x),L-t)\myr{12}
&=
-D_2 H(X(t,x),L-t)+
D_1 H(X(t,x),L-t)\frac{{\rm d}X(t,x)}{{\rm d}t}
\\
&=\left[\myl{12}
-\partial_L
+x\gamma_1(x)\partial_x
\myr{12}H(x,L)\right]_{x=X(t,x),L=L-t}=0~.
\end{equs}
If the curve ${\cal C}(x,L)$ crosses the line $L=0$ in
the $(x,L)$-plane, it does so at $t=L$, and we get
\begin{equs}
H(x,L)=H(X(L,x),0)=
X(L,x)^{-s}~,
\end{equs}
which completes the proof.
\end{proof}

\subsection{The series solution for $G(x,L)$}
\label{sec:alternate}

Throughout this section, we consider $\gamma_1$ to be a fixed
solution of (\ref{DSeqnQCD}) that exists for all $x\geq0$. We
first pick $x_0>0$ and $t_0\in{\bf R}$. These values are
arbitrary. We then introduce
\begin{equs}
T(x)=t_0+\int_{x_0}^{x}\frac{{\rm d}z}{\gamma_1(z)z}~.
\label{eqn:overLdef}
\end{equs}
We then note that
\begin{equs}
\lim_{x\to\infty}T(x)=
\left\{
\begin{array}{ll}
T_{\infty}=-\infty&\mbox{if} -1<\gamma_1(x)\leq0~~\forall
x\geq0~,\\[2mm]
T_{\infty}>-\infty&
 \mbox{if}~~\exists x_0>0\mbox{~~with~~}\gamma_1(x_0)=-1
\end{array}
\right.
~.
\label{eqn:ovL}
\end{equs}
Since $T'(x)=\frac{1}{\gamma_1(x)x}<0$,
$x\mapsto T(x)$ is an invertible map from $\real^{+}$
to $[T_{\infty},\infty)$. Moreover, we have $T(0)=\infty$ and
$T(\infty)=T_{\infty}$. The inverse map $\tilde{x}(t)$ (the
solution of $t=T(\tilde{x}(t))$) satisfies the `running
coupling equations'
\begin{equs}[3]
\frac{{\rm d}\tilde{\gamma}_1(t)}{{\rm
d}t}&=\tilde{\gamma}_1(t)+\tilde{\gamma}_1(t)^2-P(\tilde{x}(t))~~~\mbox{and}~~~
\frac{{\rm d}\tilde{x}(t)}{{\rm d}t}=\tilde{x}(t)\tilde{\gamma}_1(t)~,
\label{eqn:runn}
\end{equs}
where $\tilde{\gamma}_1(t)=\gamma_1(\tilde{x}(t))$. Since
$T(\infty)=T_{\infty}$, $\tilde{x}(t)$ diverges as $t\to
T_{\infty}$ (and so does $\tilde{\gamma_1}(t)$ if
$\gamma_1(x)=-1$ for some $x>0$).

Consider now the series solution (\ref{eqn:seriessol}). We first
introduce the functions $S_k$ such that
\begin{equs}
\gamma_k(x)=\frac{1}{k!}~x^s~S_k(T(x))
\end{equs}
for $k\geq1$. Substitution into (\ref{eqn:seriessol}) gives
\begin{equs}
S_{k}(T(x))=S_{k-1}'(T(x))
=\myl{18}\frac{{\rm d}}{{\rm
d}t}S_{k-1}(t)\myr{18}\my{|}{18}_{t=T(x)}~.
\end{equs}
Since $s=\pm1$, we find from (\ref{eqn:runn}) that
$S_1(t)=\frac{\tilde{\gamma}_1(t)}{\tilde{x}(t)^s}=-s\frac{{\rm
d}}{{\rm d}t}(\tilde{x}(t)^{-s})$. Using $s^2=1$, we find
\begin{equs}
-s\gamma_k(x)=
\frac{x^s}{k!}~
\myl{18}
\frac{{\rm d}^{k}}{{\rm d}t^{k}}
\myl{12}
\frac{1}{\tilde{x}(t)^{s}}
\myr{12}\myr{18}\my{|}{18}_{t=T(x)}
\label{eqn:solrec}
\end{equs}
for all $k\geq1$. Since $\tilde{x}(T(x))=x$, the r.h.s.\ of
(\ref{eqn:solrec}) is equal to $1$ when $k=0$, and so
\begin{equs}
G(x,L)=
1+\sum_{k=1}^{\infty}-s\gamma_k(x)L^k=
x^s
\sum_{k=0}^{\infty}
\frac{L^k}{k!}~
\myl{18}
\frac{{\rm d}^{k}}{{\rm d}t^{k}}
\myl{12}
\frac{1}{\tilde{x}(t)^{s}}
\myr{12}\myr{18}\my{|}{18}_{t=T(x)}
\label{eqn:sersolf}
~.
\end{equs}
We now fix $0<x<\infty$. Since $x$ is finite, $T(x)>T_{\infty}$,
and since the above series is a Taylor series of
$\tilde{x}(t)^{-s}$ at
$t=T(x)$, the series converges and takes the value
$\tilde{x}(t+L)^{-s}$ for small values of $L$ {\em if we assume
$z(t)$ to be analytic at $t=T(x)$}. We thus find
\begin{equs}
G(x,L)&=
\left(\frac{x}{\tilde{x}(T(x)+L)}\right)^s
\label{eqn:invpropsimp}
~,
\end{equs}
at least for sufficiently small $L$. This formula is the same as
the one of Theorem \ref{thm:transport}: 
$\tilde{x}(T(x)+t)$ solves
\begin{equs}
\frac{{\rm d}\tilde{x}(T(x)+t)}{{\rm d}t}
=\tilde{x}(T(x)+t)\gamma_1(\tilde{x}(T(x)+t))~,
\label{eqn:tildex}
\end{equs}
with initial condition $\tilde{x}(T(x))=x$ at $t=0$, and so
$\tilde{x}(T(x)+t)=X(t,x)$ by uniqueness of solutions
of (\ref{eqn:RGEt}) and (\ref{eqn:tildex}). 

The question of convergence of (\ref{eqn:sersolf}) for large $L$
depends on the (in-)finiteness of $T_{\infty}$. If
$T_{\infty}$ is finite, $\tilde{x}(t)$ has a pole as
$t\to T_{\infty}$, and (\ref{eqn:sersolf}) must diverge for
\begin{equs}
L\leq T_{\infty}-T(x)
=\int_x^{\infty}\frac{{\rm
d}z}{z\gamma_1(z)}
=L_{\infty}(x)
\end{equs}
as in Section \ref{sec:uncorrected}. However, since the series
(\ref{eqn:sersolf}) can only converge for $L$ on symmetric
intervals, it is also divergent for $L\geq-L_{\infty}(x)>0$. In
that case, we have to revert to Section \ref{sec:uncorrected} for
the validity of (\ref{eqn:invpropsimp}).

\section{Technical details for QCD}\label{sec:technicalities}

In this section, we present the technical details proving the
statements of Section \ref{sec:QCD} on QCD. In our analysis of
(\ref{DSeqnQCD}), we use mainly two tools: integral
representations of the solutions and null-clines. Our strategy is
directly inspired from \cite{reaction} and our previous work
\cite{QED} on QED.

There are two types of integral equations one can write for
(\ref{DSeqnQCD}). The first one (see also \cite{QED}) reads
\begin{equs}
\gamma_1(x)=\frac{x(1+\gamma_1(x^{\star}))}{x^{\star}}-1
-x\int_{x^{\star}}^{x}\frac{P(z)}{z^2\gamma_1(z)}{\rm d}z~.
\label{eqn:usualintegral}
\end{equs}
For the second one, we let $\gamma_1(x)$ be a solution of
(\ref{DSeqnQCD}) and $\gamma_2(x)$ be {\em any} function. Then
for all $x_0,x\in{\cal I}$, where ${\cal I}$ is the common
interval where $\gamma_1$ and $\gamma_2$ are defined, we have
\begin{equs}
\label{solutionQCD}
\gamma_1(x)-\gamma_2(x)&=
\myl{12}
\gamma_1(x_0)-\gamma_2(x_0)
\myr{12}~K[\gamma_1,\gamma_2](x_0,x)+
\int_{x}^{x_0}
\hspace{-2mm}
R[\gamma_2](y)~K[\gamma_1,\gamma_2](y,x)~{\rm d}y~,
\end{equs}
where
\begin{equs}
R[\gamma_2](x)&\equiv
\frac{{\rm d}\gamma_2(x)}{{\rm d}x} -
\frac{\gamma_2(x)
+\gamma_2(x)^2-P(x)}{x\gamma_2(x)}
~,\\
K[\gamma_1,\gamma_2](x_0,x)&\equiv
\left(
\frac{x}{x_0}
\right)
\exp\myl{18}
-
\int_{x}^{x_0}
\frac{P(z)}{z\gamma_1(z)\gamma_2(z)}
{\rm d}z
\myr{18}~.
\end{equs}
The null-clines are defined as the locations in
$(x,\gamma_1)$-plane where solutions satisfy $\gamma_1'(x)=0$).
These are given by the graph of the two functions
\begin{equs}
\gamma_c^{\pm}(x)=\frac{\pm\sqrt{1+4P(x)}-1}{2}~.
\end{equs}
In particular, as $P(0)=0$ by hypothesis H1, and
$P(x)>-\frac{1}{4}$ for $x\in[0,x^{\star}]$, the null-clines
extend at least up to the line $x=x^{\star}$ in the
$(x,\gamma_1)$-plane. On the null-clines, the second derivative
of $\gamma_1$ is given by
\begin{equs}
\gamma_1''(x)=
\frac{{\rm d}}{{\rm d}x}
\left(
f(\gamma_c^{\pm}(x),x)
\right)
=
\frac{P'(x)}{|x\gamma_c^{\pm}(x)|}~.
\label{eqn:onnull}
\end{equs}
By hypotheses H1 and H2, we have
\begin{equs}
P'(x)&= P'(0) +
\int_0^{x}
P''(z){\rm d}z
\leq P'(0)<0~~~\forall x\in[0,x^{\star}]~.
\label{eqn:Pp}
\end{equs}
Hence by (\ref{eqn:onnull}), solutions of (\ref{DSeqnQCD}) can
have at most one local maximum in the interval
$x\in[0,x^{\star}]$, and no local minimum. For further reference,
we also note that by hypotheses H1 and H2,
\begin{equs}
P(x)&= P'(0)x +
\int_0^{x}
\myl{18}
\int_0^{y}
P''(z){\rm d}z\myr{18}~{\rm d}y
\leq P'(0)x~~~\forall x\in[0,x^{\star}]~.
\label{eqn:Pconcave}
\end{equs}
In particular, we have
\begin{equs}
\min_{x\in[0,x^{\star}]}\frac{-P(x)}{x}=-P'(0)=|P'(0)|>0~,
\label{eqn:oldH3}
\end{equs}
and since $P(x^{\star})>-\frac{1}{4}$ by hypothesis on
$x^{\star}$, we have $x^{\star}<-\frac{1}{4P'(0)}<\infty$.

\subsection{Existence, uniqueness and properties of the
asymptotic freedom solution}

We can now establish the existence, uniqueness and properties of
the solution with asymptotic freedom.
\begin{theorem}
\label{thm:existenceanduniquenessatzero}
Under the hypotheses {\rm H1} and {\rm H2}, there exists a unique
value $\gamma_1^{\star}(x^{\star})$ such that the corresponding
solution $\gamma_1^{\star}(x)$ of (\ref{DSeqnQCD}) exists for all
$x\in[0,x^{\star}]$ and satisfies
${\displaystyle\lim_{x\to0}}\gamma_1^{\star}(x)=0$. 
\end{theorem}

\begin{proof}
We first prove that the solution is unique. Namely, assume
{\em ab absurdum} that $\gamma_1$ and $\gamma_2$ are two
solutions of (\ref{DSeqnQCD}) on $[0,x^{\star}]$ satisfying
${\displaystyle\lim_{x\to0}}\gamma_i(x)=0$. Since
$\gamma_i(x^{\star})<0$ and
${\displaystyle\lim_{x\to0}}\gamma_i(x)=0$, we necessarily have
$\gamma_i(x)\geq\gamma_c^{+}(x)$ for all $x\in[0,x^{\star}]$.
Since $\gamma_2$ is a solution, we can apply (\ref{solutionQCD})
with $R[\gamma_2]=0$, and using (\ref{eqn:oldH3}), we get
\begin{equs}
|\gamma_1(x)-\gamma_2(x)|\geq|\gamma_1(x^{\star})-
\gamma_2(x^{\star})|
\frac{x}{x^{\star}}
\exp\myl{18}
\int_{x}^{x^{\star}}
\frac{|P'(0)|}{\gamma_c^{+}(z)^2}
{\rm d}z
\myr{18}~.
\label{eqn:rhsdi}
\end{equs}
Since $\gamma_c^{+}(z)=P'(0)z+{\cal O}(z^2)$ as $z\to0$, the
r.h.s.\ of (\ref{eqn:rhsdi}) diverges as $x\to0$, which
contradicts ${\displaystyle\lim_{x\to0}}\gamma_i(x)=0$. This
shows that there can be at most one solution of (\ref{DSeqnQCD})
satisfying ${\displaystyle\lim_{x\to0}}\gamma_1(x)=0$.

To prove the existence of this solution, we define the two sets
\begin{equs}
{\rm I}_1 &= \{ \gamma_1(x^{\star})
\in]\gamma_c^{+}(x^{\star}),0[
~\mbox{s.t.}~\exists
x_{\rm min}\in]0,x^{\star}[\mbox{ with }\gamma_1(x_{\rm min})=0\}~,\\
{\rm I}_2 &= \{ \gamma_1(x^{\star})\in]\gamma_c^{+}(x^{\star}),0[~\mbox{s.t.}~\exists
x_{1}\in]0,x^{\star}[\mbox{ with }\gamma_1(x_1)=\gamma_c^{+}(x_1)\}~.
\end{equs}
We will prove in Propositions \ref{prop:Ione} and \ref{prop:Itwo}
below that these sets are non-empty. Continuity of solutions
w.r.t.\ initial conditions imply that they are open, while
(\ref{solutionQCD}) with $R[\gamma_2]=0$ shows that solutions are
{\em ordered}, which imply that each ${\rm I}_i$ is a single
interval. Now since ${\rm I}_1$ and ${\rm I}_2$ are disjoint open
intervals, there exist at least one initial condition
$\gamma_1^{\star}(x^{\star})$ that is in neither sets, and hence
the corresponding solution $\gamma_1^{\star}(x)$ satisfies
$\lim_{x\to0}\gamma_1^{\star}(x)=0$.
\end{proof}

To establish that ${\rm I}_1$ is non-empty, we show that initial
conditions at $x=x^{\star}$ sufficiently close to the $x$-axis
necessary cross it at some $x_{\rm min}<x^{\star}$.
\begin{proposition}
\label{prop:Ione}
For all $\gamma_0\in[P'(0)x^{\star},0[$, there exists $x_{\rm
min}\in]0,x^{\star}[$ such that the solution of (\ref{DSeqnQCD})
only exists on $x\in[x_{\rm min},x^{\star}]$ and $\gamma_1(x_{\rm
min})=0$.
\end{proposition}
\begin{proof}
Pick $\gamma_2(x)=P'(0) x$. Then by (\ref{eqn:Pconcave}), we have
\begin{equs}
R[\gamma_2](x)
=
\frac{1}{x^2P'(0)}
\int_0^{x}
\myl{18}
\int_0^{y}
P''(z){\rm d}z\myr{18}~{\rm d}y
>0~~~\forall x\in[0,x^{\star}].
\end{equs}
Applying (\ref{solutionQCD}), we get
\begin{equs}
\gamma_1(x)=P'(0)x+
\int_{x}^{x^{\star}}
\hspace{-2mm}
R[\gamma_2](y)~K[\gamma_1,\gamma_2](y,x)~{\rm d}y> P'(0)x~~~
\forall x<x^{\star}~.
\label{eqn:abovepp}
\end{equs}
In particular, $\gamma_1$ cannot cease to exist by diverging to
$-\infty$ at a finite $x<x^{\star}$. We also see from
(\ref{eqn:abovepp}) that there exists $x_0<x^{\star}$ with 
$\gamma_1(x_0)>P'(0)x_0$. Using $R[\gamma_2]>0$,
(\ref{eqn:oldH3}) and $\gamma_1(x)\geq P'(0) x$ in
(\ref{solutionQCD}), we find
\begin{equs}
\gamma_1(x)&\geq P'(0) x+
\left(\gamma_1(x_0)-P'(0) x_0\right)
\frac{x}{x_0}
\exp\myl{18}
\int_{x}^{x_0}
\frac{-1}{P'(0) z^2}
{\rm d}z
\myr{18}
\\
&\geq 
\frac{x}{x_0}\left(
P'(0) x_0
+
(\gamma_1(x_0)-P'(0) x_0)
\ed^{\frac{1}{|P'(0)| x}-\frac{1}{|P'(0)| x_0}}
\right)~.
\end{equs}
The proof is completed since $\gamma_1(x_0)-P'(0)x_0>0$.
\end{proof}

Before proving that the interval ${\rm I}_2$ is
non-empty, we can establish an additional property of the
asymptotically free solution $\gamma_1^{\star}$:
\begin{corollary}
\label{cor:asympfreecor}
The solution $\gamma_1^{\star}$ of Theorem
\ref{thm:existenceanduniquenessatzero} satisfies
\begin{equs}
\gamma_c^{+}(x)<\gamma_1(x)< P'(0)x~.
\label{eqn:sandrap}
\end{equs}
for all $x\in]0,x^{\star}]$.
\end{corollary}
\begin{proof}
To get the lower bound in (\ref{eqn:sandrap}), we recall that
solutions can have at most one local maximum in $]0,x^{\star}]$
and no local minimum. Since
${\displaystyle\lim_{x\to0}}\gamma_1^{\star}(x)=0$,
$\gamma_1(x)>\gamma_c^{+}(x)$ for all $x\in]0,x^{\star}]$. The
upper bound follows immediately from Proposition \ref{prop:Ione}.
\end{proof}

We now show that initial conditions sufficiently close (but
above) the null-cline $\gamma_c^{+}(x^{\star})$ necessarily cross
it at some $x<x^{\star}$.

\begin{proposition}
\label{prop:Itwo}
There exist $\epsilon_1\ll1$ sufficiently small such that the
solution of (\ref{DSeqnQCD}) with
$\gamma_1(x^{\star})=\gamma_c^{+}(x^{\star})(1-\epsilon_1)$
satisfies $\gamma_1(x)=\gamma_c(x)$ for some
$0<x<x^{\star}$.
\end{proposition}

\begin{proof}
Let $0<\epsilon_1<\frac{1}{2}$ and
$\gamma_1(x^{\star})=\gamma_c^{+}(x^{\star})(1-\epsilon_1)$.
By continuity of solutions and since
$\gamma_c^{+}(x^{\star})<\gamma_1(x^{\star})<\gamma_1(x^{\star})
(1-2\epsilon_1)$, there exists $0<x_1<x^{\star}$ such that
\begin{equs}
\left\{
\begin{array}{l}
\gamma_c^{+}(x^{\star})<\gamma_1(x)<
\gamma_c^{+}(x^{\star})(1-2\epsilon_1)~~~\forall
x\in]x_1,x^{\star}]~~~\mbox{and}\\[2mm]
\gamma_1(x_1)=\gamma_c^{+}(x^{\star})
~~\mbox{ or }~~\gamma_1(x_1)=\gamma_c^{+}(x^{\star})(1-2\epsilon_1)
\end{array}
\right.
\label{eqn:apriori}~.
\end{equs}
Using these inequalities and hypothesis H2, we get
\begin{equs}
\frac{{\rm d}\gamma_1}{{\rm d}x}\geq
-\frac{2\epsilon_1}{1-2\epsilon_1}
\frac{P(x^{\star})}{\gamma_c^{+}(x^{\star})}
~~~\forall x\in[x_1,x^{\star}]~,
\end{equs}
which, upon integration, gives
\begin{equs}
\gamma_1(x)\leq
\gamma_c^{+}(x^{\star})(1-\epsilon_1)+
\frac{2\epsilon_1}{1-2\epsilon_1}
\frac{P(x^{\star})}{\gamma_c^{+}(x^{\star})}
\ln\left(
\frac{x^{\star}}{x}
\right)
~~~\forall x\in[x_1,x^{\star}]~.
\label{eqn:derivedapriori}
\end{equs}
Let now $x_2$ be the value at which the r.h.s.\ of
(\ref{eqn:derivedapriori}) attains
$\gamma_c^{+}(x^{\star})(1-2\epsilon_1)$, namely
\begin{equs}
x_2=x^{\star}
\exp\left(
\frac{-\gamma_c^{+}(x^{\star})^2(1-2\epsilon_1)}{2|P(x^{\star})|}
\right)<x^{\star}~.
\end{equs}
We now consider the two alternatives $x_2<x_1$ and $x_2\geq x_1$.

Assume first that $x_2<x_1$. Then (\ref{eqn:derivedapriori})
shows that
$\gamma_1(x_1)<\gamma_c^{+}(x^{\star})(1-2\epsilon_1)$, and thus
by definition of $x_1$ (see (\ref{eqn:apriori})), we have
$\gamma_1(x_1)=\gamma_c^{+}(x^{\star})$ and since by H2,
$\gamma_c^{+}(x)$ increases as $x\to0$, there exists an
$x\in[x_1,x^{\star}]$ with  $\gamma_1(x)=\gamma_c^{+}(x)$.

Consider now the other possible case, namely $x_2\geq x_1$. Since
now $\gamma_c^{+}(x^{\star})\leq\gamma_1(x)\leq
\gamma_c^{+}(x^{\star})(1-2\epsilon_1)$ for all
$x\in[x_2,x^{\star}]$, and we conclude from
(\ref{eqn:derivedapriori}) that
\begin{equs}
\gamma_1(x)-\gamma_c^{+}(x)\leq
\gamma_c^{+}(x^{\star})-\gamma_c^{+}(x)
-\epsilon_1\gamma_c^{+}(x^{\star})+
\frac{2\epsilon_1}{1-2\epsilon_1}
\frac{P(x^{\star})}{\gamma_c^{+}(x^{\star})}
\ln\left(
\frac{x^{\star}}{x}
\right)~.
\label{eqn:finif}
\end{equs}
The proof is completed by noting that for $x$ very close to (but
strictly less than) $x^{\star}$,
$\gamma_c^{+}(x^{\star})-\gamma_c^{+}(x)$ is negative, while the
last two (positive terms) in (\ref{eqn:finif}) can be made
arbitrarily small by picking $\epsilon_1$ small enough.
\end{proof}

\subsection{Behavior towards $x=0$ of non-asymptotically free
solutions}

We first consider solutions of (\ref{DSeqnQCD}) corresponding to
initial conditions
$\gamma_1(x^{\star})<\gamma_1^{\star}(x^{\star})$. We have the
following result.
\begin{proposition}
Let $\gamma_0<\gamma_1^{\star}(x^{\star})$. The corresponding
solution $\gamma_1(x)$ of (\ref{DSeqnQCD}) exists for all
$x\in[0,x^{\star}]$ and satisfies $\gamma_1(x)=-1+{\cal
O}(x\ln(x))$ as $x\to0$.
\end{proposition}

\begin{proof}
Note first that since $\gamma_0<\gamma_1^{\star}(x^{\star})$, 
there always exists $x_0\leq x^{\star}$ such that
$\gamma_1(x_0)\leq \gamma_c^{+}(x_0)$. Namely, if this does not
already hold at $x^{\star}$, the proof of Theorem
\ref{thm:existenceanduniquenessatzero} shows that it will
eventually hold at some smaller value of $x$. Since solutions can
have at most one local maximum and no local minimum in
$[0,x^{\star}]$, we then have
\begin{equs}
\min(-1,\gamma_1(x_0))\leq\gamma_1(x)\leq\gamma_c^{+}(x_0)~~~~~
\forall x\in[0,x_0]~.
\end{equs}
In particular, these solutions exist for all values of
$x\in[0,x^{\star}]$. We then apply (\ref{eqn:usualintegral}) and
get
\begin{equs}
-1+Cx
+\frac{x}{|\gamma_c^{+}(x_0)|}
\int_x^{x_0}
\frac{-P(z)}{z^2}{\rm d}z
\leq \gamma_1(x)
\leq
-1+Cx
+\frac{x}{\max(1,|\gamma_c^{+}(x_0)|)}
\int_x^{x_0}
\frac{-P(z)}{z^2}{\rm d}z~,
\end{equs}
where $C=\frac{1+\gamma_1(x_0)}{x_0}$. By hypotheses H1 and H2
(see also (\ref{eqn:Pconcave}), $P(x)=P'(0)x+{\cal O}(x^2)$
it is then straightforward to prove that $\gamma_1(x)=-1+{\cal
O}(x(1+\ln(x)))$ as $x\to0$, which completes the proof.
\end{proof}

We now consider solutions of (\ref{DSeqnQCD}) corresponding to
initial conditions $\gamma_1(x^{\star})$ in the interval
$]\gamma_1^{\star}(x^{\star}),0[$. We have the following result.
\begin{proposition}
\label{prop:doublevalued}
Let $\gamma_0\in]\gamma_1^{\star}(x^{\star}),0[$. The corresponding
solution $\gamma_1(x)$ of (\ref{DSeqnQCD}) satisfies
$\gamma_1(x_{\rm min})=0$ for some $x_{\rm min}\in]0,x^{\star}[$.
It can then be continued and enters the first quadrant (becoming
double-valued), and thus satisfies $\gamma_1(x_0)>0$ for 
$x>x_{\rm min}$.
\end{proposition}

\begin{proof}
Let $x_{\rm min}\geq0$ be the minimal value such that $\gamma_1(x)$
exists $\forall x\in[x_{\rm min},x^{\star}]$. We claim that
$x_{\rm min}>0$ and $\gamma_1(x_{\rm min})=0$. Namely, since
$\gamma_1(x^{\star})>\gamma_1^{\star}(x^{\star})$, we have by
(\ref{solutionQCD}) with $R[\gamma_2]=0$ that
\begin{equs}
\gamma_1(x)=\gamma_1^{\star}(x)+(\gamma_1(x^{\star})-
\gamma_1^{\star}(x^{\star}))
\frac{x}{x^{\star}}
\exp\myl{18}
\int_{x}^{x^{\star}}
\frac{-P(z)}{z\gamma_1^{\star}(z)\gamma_1(z)}
{\rm d}z
\myr{18}>\gamma_1^{\star}(x)
\label{eqn:happa}
\end{equs}
for all $x\in[x_{\min},x^{\star}]$. Assuming {\em ab
absurdum} that $x_{\rm min}=0$ and $\gamma_1(x)<0$ for all 
$x\in[0,x^{\star}]$ leads to a contradiction, for then we would
have $\gamma_1(x)>\gamma_1^{\star}(x)\geq\gamma_c^{+}(x)$ for all
$x\in[0,x^{\star}]$, and using (\ref{eqn:oldH3}) and
(\ref{eqn:happa}), we get
\begin{equs}
\gamma_1(x)\geq\gamma_1^{\star}(x)+(\gamma_1(x^{\star})-
\gamma_1^{\star}(x^{\star}))
\frac{x}{x^{\star}}
\exp\myl{18}
\int_{x}^{x^{\star}}
\frac{|P'(0)|}{\gamma_c^{+}(z)^2}
{\rm d}z
\myr{18}
\end{equs}
which goes to $+\infty$ as $x\to0$. So $x_{\rm min}>0$ and
$\gamma_1(x_{\rm min})=0$.
Although (\ref{DSeqnQCD}) is singular at $\gamma_1(x_{\rm
min})=0$, these solutions can be continued in the first quadrant
by reverting to the so-called `running coupling' formulation of
(\ref{DSeqnQCD}) (see also (\ref{eqn:runn}) and \cite{QED}).
Namely, we introduce a new independent variable $t$, and write 
$x=X(t)$ and $\gamma_1(X(t))=\tilde{\gamma}_1(t)$, getting
\begin{equs}[3]
\frac{{\rm d}\tilde{\gamma}_1(t)}{{\rm
d}t}&=\tilde{\gamma}_1(t)+\tilde{\gamma}_1(t)^2-P(X(t))~~~~~ &
\tilde{\gamma}_1(t_0)&=0~,\\
\frac{{\rm d}X(t)}{{\rm d}t}&=X(t)\tilde{\gamma}_1(t) & X(t_0)&=x_{\rm
min}~.
\end{equs}
These equations are {\em not} singular at $\tilde{\gamma}_1=0$, and thus
solutions will exist (at least locally around $t=t_0$). Since
$P(x)<0$, the solution to these equations will satisfy
$\tilde{\gamma}_1(t)=\gamma_0>0$ and $X(t)=x_0>x_{\rm min}$ for
some finite $t>t_0$. 
\end{proof}

\subsection{Behavior as $x\to\infty$}

We first show that solutions in the first quadrant are global,
and satisfy appropriate estimates as $x\to\infty$.

\begin{proposition}
\label{prop:doublevaluedinfinity}
Let $\gamma_1(x_0)>0$, and assume $P(x)$ satisfies {\rm H3}. The
corresponding solution $\gamma_1(x)$ exists for all $x\geq
x_{0}$, and satisfies
\begin{equs}
0<
x S_{P}(x_0,x)
\leq
\gamma_1(x)
\leq
xS_{P}(x_0,x)+
\frac{x}{x_0}-1
\end{equs}
for all $x\geq x_0$.
\end{proposition}

\begin{proof}
We first note that
\begin{equs}
\frac{1}{2}
\frac{{\rm d}}{{\rm d}x}
\myl{14}
\gamma_1(x)^2
\myr{14}=
\gamma_1(x)
\frac{{\rm d}\gamma_1}{{\rm d}x}=\frac{\gamma_1(x)}{x}
+\frac{\gamma_1(x)^2-P(x)}{x}
\geq\frac{\gamma_1(x)^2-P(x)}{x}~.
\end{equs}
By integration, we find
\begin{equs}
\gamma_1(x)\geq
x\sqrt{
\frac{\gamma_1(x_0)^2}{x_0^2}
+2
\int_{x_0}^{x}
\frac{-P(z)}{z^3}{\rm d}z
}
=x S_{P}(x_0,x)
>0~.
\label{eqn:tratra}
\end{equs}
This shows that solutions cannot cease to exist by reaching
$\gamma_1(x)=0$ at some $x>x_{0}$. Inserting
(\ref{eqn:tratra}) into (\ref{eqn:usualintegral}) gives
\begin{equs}
\gamma_1(x)&\leq 
\frac{x(1+\gamma_1(x_0))}{x_0}-1
+
x
\int_{x_0}^{x}
\frac{-P(z)}{z^3S_{P}(x_0,z)}{\rm d}z
=x S_{P}(x_0,x)+\frac{x}{x_0}-1~,
\label{eqn:rhsbpgz}
\end{equs}
since
\begin{equs}
\frac{{\rm d}S_{P}(x_0,x)}{{\rm d}x}
=-\frac{P(x)}{x^3S_P(x_0,x)}~.
\end{equs}
The proof is completed since (\ref{eqn:rhsbpgz}) shows that
solutions cannot cease to exist by diverging to $\infty$ at a
finite $x\geq x_0$ either.
\end{proof}

We now turn to the fate of any type of solutions of
(\ref{DSeqnQCD}) with $\gamma_1(x^{\star})<0$ as $x\to\infty$.
Our first result is that these solutions are global, {\em i.e.},
they can be extended as $x\to\infty$. In particular, the
asymptotically free $\gamma_1^\star(x)$ is global.

\begin{proposition}
\label{prop:negglobal}
Let $\gamma_1(x^{\star})<0$ and assume $P(x)$ satisfies {\rm H3}.
The corresponding solution of (\ref{DSeqnQCD}) exists for all
$x\geq x^{\star}$, and satisfies
\begin{equs}
-xS_P(x^{\star},x)\leq \gamma_1(x)
\leq 
\max\left(\gamma_1(x^{\star}),
\sup_{z\in[x^{\star},x]}
\frac{\sqrt{1+4P(z)}-1}{4}\right)<0
\end{equs}
for all $x\geq x^{\star}$. In particular, if ${\cal
D}(P)<\infty$, solutions grow at most linearly as $x\to\infty$.
\end{proposition}
\begin{proof}
For the lower bound, note first that, as in the proof of
Proposition \ref{prop:doublevalued}, we have
\begin{equs}
\frac{1}{2}
\frac{{\rm d}}{{\rm d}x}
\myl{14}\gamma_1(x)^2\myr{14}=
\gamma_1(x)\frac{{\rm d}\gamma_1}{{\rm d}x}=\frac{\gamma_1(x)}{x}
+\frac{\gamma_1(x)^2-P(x)}{x}
\leq\frac{\gamma_1(x)^2-P(x)}{x}~,
\end{equs}
which gives $\gamma_1(x)\geq-xS_P(x^{\star},x)$ upon integration.
This shows that solutions cannot diverge to $-\infty$ at a finite
$x>x^{\star}$. Now suppose {\em ab absurdum} that there exists
$x_{\rm max}<\infty$ such that $\gamma_1(x_{\rm max})=0$. By
hypothesis H1-H3, we have
$-\frac{1}{4}<\sup_{x\in[x^{\star},x_{\rm max}]}P(x)<0$, and thus
\begin{equs}
\gamma_{\rm max}=
\max\left(\gamma_1(x^{\star}),
\sup_{z\in[x^{\star},x_{\rm max}]}
\frac{\sqrt{1+4P(z)}-1}{4}\right)
\end{equs}
satisfies $-\frac{1}{4}<\gamma_{\rm max}<0$. Note then that
$f(\gamma_1,x)$ is strictly negative along the
$\gamma_1=\gamma_{\max}$ line since
\begin{equs}
\sup_{x\in[x^{\star},x_{\rm max}]}
xf(\gamma_{\rm max},x)
\leq
\left(1+\gamma_{\rm max}-\frac{\delta}{\gamma_{\rm max}}\right)
\leq-\frac{1}{4}~.
\end{equs}
Since $\gamma_1(x^{\star})\leq \gamma_{\rm max}$, this shows that
$\gamma_1(x)\leq\gamma_{\rm max}$ for all $x\in[x^{\star},x_{\rm
max}]$, contradicting the {\em ab absurdum} assumption. Hence
solutions exist globally as $x\to\infty$.
\end{proof}

Our second result concern the asymptotics of some of these
solutions as $x\to\infty$. Namely, we can estimate the growth of
solutions that are somewhere less than $-1$.
\begin{proposition}
\label{prop:growthunder}
Assume $P(x)$ satisfies {\rm H3} and $\gamma_1(x^{\star})<0$. If
the corresponding solution of (\ref{DSeqnQCD}) satisfies
$\gamma_1(x_0)\leq -1$ for some $x_0\geq x^{\star}$, then
\begin{equs}
-x S_{P}(x_0,x)
\leq
\gamma_1(x)&\leq
-xS_{P}(x_0,x)
+\frac{x}{x_0}-1<0
\end{equs}
for all $x\geq x_0$.
\end{proposition}

\begin{proof}
The lower bound is already contained in Proposition
\ref{prop:negglobal}. The upper bound then follows immediately
from the lower bound and the integral formulation
(\ref{eqn:usualintegral}). 
\end{proof}

The condition $\gamma_1(x_0)\leq-1$ is essential in Proposition
\ref{prop:growthunder} to guarantee that the upper bound is
indeed negative. We now give possible scenarios that guarantee
solutions indeed reach $\gamma_1=-1$.

\begin{proposition}
\label{prop:minusonerap}
Assume one of the two following statements holds:
\begin{enumerate}
\item $-1<\gamma_1(x^{\star})<0$ and $P(x)$ satisfies S1 or S3,
\item $-1<\gamma_1(x^{\star})\leq\gamma_1^{\star}(x^{\star})$ and
$P(x)$ satisfies S2.
\end{enumerate}
Then there exists $x_0>x^{\star}$ such that the corresponding
solution $\gamma_1(x)$ satisfies $\gamma_1(x_0)=-1$.
\end{proposition}

\begin{proof}
We consider the alternative 1. first. Note that S1 implies S3, so
we can use hypothesis S3 only. As shown in Proposition
\ref{prop:negglobal}, any solution starting at
$\gamma_1(x^{\star})<0$ exists for all values of $x\geq
x^{\star}$. In particular, $\gamma_1(x_{c})=\gamma_{\rm min}<0$.
If $\gamma_{\rm min}\leq -1$, the proof is completed. If
$-1<\gamma_{\rm min}<0$, we assume {\em ab absurdum} that
$\gamma_1(x)>-1$ for all $x\in[x_{c},x_{d}]$. Note that there
cannot be an $x\in[x_{c},x_{d}]$ such that $\gamma_1(x)=0$, hence
$-1<\gamma_1(x)<0$ for all $x\in[x_{c},x_{d}]$, and we have
\begin{equs}
\frac{{\rm d}}{{\rm d}x}(\gamma_1(x)^2)=
2\frac{\gamma_1(x)+\gamma_1(x)^2+\frac{1}{4}}{x}
-\frac{4P(x)+1}{2x}
\geq -\frac{4P(x)+1}{2x}
\end{equs}
for all $x\in[x_{c},x_{d}]$. Upon integration, we
thus find that
\begin{equs}
\gamma_1(x)\leq -\sqrt{\gamma_{\rm min}^2-
\int_{x_{c}}^{x}
\frac{1+4P(z)}{2z}
{\rm d}z
}\leq-\sqrt{1+\gamma_{\rm min}^2}<-1
\end{equs}
by hypothesis S3, which is a contradiction.

Consider then the alternative 2. Under hypothesis S2, we can
extend (\ref{eqn:oldH3}) to get $P(x)\leq P'(0)x$ for all
$x\in[0,x_c]$. Consider now $\gamma_2(x)=P'(0)x$. We have
$R[\gamma_2](x)>0$ for all $x\in[0,x_c]$. Thus, since
$\gamma_1(x^{\star})\leq
\gamma_1^{\star}(x^{\star})$ and
$\gamma_1^{\star}(x^{\star})<P'(0)x^{\star}$ by Corollary
\ref{cor:asympfreecor}, we find from (\ref{solutionQCD}) that
\begin{equs}
\gamma_1(x)&=
P'(0)x-
\my{|}{12}
\gamma_1(x^{\star})-\gamma_2(x^{\star})
\my{|}{12}~K[\gamma_1,\gamma_2](x^{\star},x)-
\int_{x^{\star}}^{x}
\hspace{-2mm}
R[\gamma_2](y)~K[\gamma_1,\gamma_2](y,x)~{\rm d}y
\leq P'(0)x
\end{equs}
for all $x\in[x^{\star},x_c]$. The proof is completed since 
$x_c>-\frac{1}{P'(0)}$ by hypothesis S2, and thus
$\gamma_1(x_c)\leq P'(0)x_c\leq-1$.
\end{proof}

We conclude this section by showing that ${\cal D}(P)<\infty$
implies that all solutions that are either positive or go below
$\gamma_1=-1$ have a finite slope as $x\to\infty$.

\begin{proposition}
Assume ${\cal D}(P)<\infty$ and $\gamma_1(x_0)\leq-1$ or 
$\gamma_1(x_0)>0$, there exists $s>0$ such that
\begin{equs}
\lim_{x\to\infty}\frac{\gamma_1(x)}{x}=
\left\{
\begin{array}{rl}
-s<0 & \mbox{if}~~\gamma_1(x_0)\leq -1\\[2mm]
 s>0 & \mbox{if}~~\gamma_1(x_0)> 0
\end{array}\right.
~.
\end{equs}
If $\gamma_1(x_0)\leq-1$, the convergence towards the limit is
given by
\begin{equs}
\my{|}{14}\frac{\gamma_1(x)}{x}+s\my{|}{14}
\leq C\int_{x}^{\infty}
\frac{-P(z)}{z^3}{\rm d}z~.
\label{eqn:linearatinfoneagain}
\end{equs}
If ${\cal D}(P)<\infty$ and $\gamma_1(x_0)>0$, then
(\ref{eqn:linearatinfone}) also hold, with $-s$ replaced by $s$.

\end{proposition}
\begin{proof}
Note that from Proposition \ref{prop:doublevaluedinfinity} and
\ref{prop:growthunder}, the hypothesis ${\cal D}(P)<\infty$
implies that all solutions under consideration here satisfy
\begin{equs}
c_1~x\leq
|\gamma_1(x)|
\leq c_2~x
\label{eqn:sand}
\end{equs}
for some $c_1,c_2>0$ and all $x\geq x_0$. The integral
formulation (\ref{eqn:usualintegral}) then gives
\begin{equs}
\frac{\gamma_1(x)}{x}=\frac{1+\gamma_1(x_0)}{x_0}
-\frac{1}{x}
-\int_{x_0}^{x}\frac{P(z)}{z^2\gamma_1(z)}{\rm d}z~,
\label{eqn:limitlimitohmylimit}
\end{equs}
from which the proof follows immediately, since the r.h.s.~of
(\ref{eqn:limitlimitohmylimit}) converges by (\ref{eqn:sand}) and
the hypothesis ${\cal D}(P)<\infty$.
\end{proof}

\subsection{The confinement solution}
\label{sec:confsol}

In this section, we consider $P(x)$ satisfying the hypotheses
H1,H2 and S4. In particular, recall that we assume the existence
of $x_r>0$ such that $P(x_r)=-\frac{1}{4}$ and
$P(x)>-\frac{1}{4}$ for all $x>x_r$ and $P(x)$ tending to a
finite limit as $x\to\infty$. Any solution of (\ref{DSeqnQCD})
that satisfies
\begin{equs}
\lim_{x\to\infty}\gamma_1(x)=
\gamma_{\infty}\equiv-\frac{1+\sqrt{1+4P_{\infty}}}{2}~,
\label{eqn:limimi}
\end{equs}
needs to solve the integral equation obtained by taking the
(formal) limit $x^{\star}\to\infty$ in (\ref{eqn:usualintegral}),
namely
\begin{equs}
\gamma_1^c(x)=-1
+x\int_{x}^{\infty}\frac{P(z)}{z^2\gamma_1^c(z)}{\rm d}z
=
-1
+\int_{1}^{\infty}\frac{P(xt)}{t^2\gamma_1^c(xt)}{\rm d}t
~.
\label{eqn:integralatinf}
\end{equs}
Defining
\begin{equs}
{\cal T}[h](x)
=
\int_{1}^{\infty}\frac{P(xt)-P_{\infty}}{t^2(\gamma_{\infty}+h(xt))}{\rm d}t
-
\frac{P_{\infty}}{\gamma_{\infty}}
\int_{1}^{\infty}\frac{h(xt)}{t^2(\gamma_{\infty}+h(xt))}{\rm d}t
~,
\end{equs}
we see that any confinement solution can be written as
$\gamma_1^c(x)=\gamma_{\infty}+h(x)$ where $h(x)$ satisfies $h(x)={\cal
T}[h](x)$. Consider then ${\cal B}_{x_0}$ the Banach space
obtained by completing the space of ${\cal
C}_0^{\infty}([x_0,\infty),{\bf R})$ functions under the
norm
\begin{equs}
\|f\|_{x_0}\equiv \sup_{x\geq x_0}|f(x)|+x|f'(x)|~.
\end{equs}
Since
$\displaystyle\lim_{x\to\infty}P(x)-P_{\infty}=\lim_{x\to\infty}xP'(x)=0$,
$\|P-P_{\infty}\|_{x_0}$ can be made as small as one likes by taking
$x_0>x_r$ large enough. Standard arguments then show that ${\cal
T}$ is a {\em contraction} in a ball of positive radius $\rho<1$
centered at $0$ in ${\cal B}$, which shows there exists a unique
$h\in{\cal B}$ solving $h={\cal T}[h]$. Since $h(x)$ is regular,
$\gamma_1^{c}(x)=\gamma_{\infty}+h(x)$ solves (\ref{DSeqnQCD}) for all $x\geq
x_0$. We now remark that $\gamma_1^{c}(x)$ will satisfy
(\ref{eqn:integralatinf}) as long as it exists when $x$ decreases
below $x_0$. However, it can only cease to exist if it satisfies
$\gamma_1(x_{\rm min})=0$ for some $x_{\rm min}>0$, for the
r.h.s.\ of (\ref{DSeqnQCD}) is negative for large negative
$\gamma_1$. Assuming it can be continued up to $x=0$, we can
conclude from Theorem (\ref{thm:completechar}) that either
$\gamma_1^{c}(0)=0$ or $\gamma_1^{c}(0)=-1$. Finally, if
$P'(x)>0$ for all $x>x_{\rm max}$, then the lower nullcline
\begin{equs}
\gamma_1^{-}(x)=
-
\frac{\sqrt{1+4P(x)}+1}{2}
\end{equs}
is decreasing towards $\gamma_{\infty}$. If there was an
$x_0>x_{\rm max}$ such that
$\gamma_1^{c}(x_0)=\gamma_1^{-}(x_0)$, then $\gamma_{1}^{c}(x)$
would enter a region of strictly positive derivatives w.r.t.\
$x$, and hence we would get
$\gamma_1^{c}(x)>\gamma_1^{-}(x_0)>\gamma_{\infty}$,
contradicting (\ref{eqn:limimi}). Similarly, if there was an
$x_0>x_{\rm max}$ such that $\gamma_1^{c}(x_0)=\gamma_{\infty}$,
then $\gamma_1^{c}(x)$ would enter a region of strictly negative
derivative w.r.t.\ $x$, and hence would satisfy
$\gamma_1^{c}(x)<\gamma_{\infty}$ for all $x>x_0$, contradicting
(\ref{eqn:limimi}) again. We thus find that under the
monotonicity assumption on $P(x)$, we have
\begin{equs}
\gamma_{\infty}\equiv-\frac{1+\sqrt{1+4P_{\infty}}}{2}
<
\gamma_1(x)
<
-\frac{1+\sqrt{1+4P(x)}}{2}
\end{equs}
for all (finite) $x>x_{\rm max}$. This concludes the proof of
Theorem \ref{thm:confinement}.

\section{Post Scriptum: QED\cite{QED} revisited}\label{sec:QED}

In a previous publication (see \cite{QED}), we have considered
the initial value problem
\begin{equs}
\label{DSeqn}
\frac{{\rm d}\gamma_1(x)}{{\rm d}x} = 
f_s(\gamma_1(x),x)\equiv
\frac{\gamma_1(x)
+\gamma_1(x)^2-P(x)}{sx\gamma_1(x)}~,~~~
\gamma_1(x_0)=\gamma_0>0~.
\end{equs}
with $0<x_0<1$ fixed and $P$ a ${\cal C}^2$ function on
$[0,\infty)$, positive on $(0,\infty)$, with $P(0)=0$ and
$P'(0)\neq0$. Defining the following (possibly infinite)
quantity
\begin{equs}
{\cal D}(P)=\int_{x_0}^{\infty}\frac{-P(z)}{z^3}{\rm d}z~,
\end{equs}
we showed that ${\cal D}(P)$ was intimately linked to the
behavior/existence as $x\to\infty$ of solutions starting with
$\gamma_1(x_0)>0$. Namely, if ${\cal D}(P)<\infty$, we showed
that there was a smallest (non-zero) value
$\gamma_1^{\star}(x_0)$ separating global solutions (with
$\gamma_1(x_0)\geq\gamma_1^{\star}(x_0)$) from solutions that
cannot be continued to $x=\infty$. We also showed that despite
the singular nature of (\ref{DSeqn}) as $x\to0$ and/or
$\gamma_1\to0$, all solutions of (\ref{DSeqn}) could be continued
for all $x\in[0,x_0]$, and would approach $(0,0)$ while
satisfying for all $x\in[0,x_0]$ the bound
\begin{equs}
\gamma_1(x)\leq
\my{\{}{24}
\begin{array}{ll}
C_b~x &\mbox{if }~~s<1 \\[1mm]
C_b~x~|\ln(x)| & \mbox{if }~~s=1\\[1mm]
C_b~x^{1/s} & \mbox{if }~~s>1
\end{array}~,
\label{eqn:gnagna}
\end{equs}
for some constant $C_b=C_b(s,\gamma_0)$. In Lemma
\ref{lem:ratamiaou} below, we will improve the above to get
linear bounds in all cases (i.e. for $s\geq1$ as well). In the
mean time, we define, for any (finite) integer $p\geq2$, the
truncation of the (divergent) series solution:
\begin{equs}
\gamma_{2,p}(x)=P'(0)x+\sum_{n=2}^{p}a_n~x^n~.
\end{equs}
The coefficients $\{a_n\}_{n=2\ldots p}$ can be found recursively
by imposing $|R[\gamma_{2,p}](x)|\leq C x^{p-1}$ as $x\to0$,
where the remainder map $R$ is defined by
\begin{equs}
R[\gamma_2](x)&\equiv
\frac{{\rm d}\gamma_2(x)}{{\rm d}x} -
\frac{\gamma_2(x)
+\gamma_2(x)^2-P(x)}{sx\gamma_2(x)}
\end{equs}
for $\gamma_2(x)$ {\em any} function on $[0,x_0]$. Note that by
choosing $x_0$ sufficiently small, we can ensure that
$\gamma_{2,p}(x)>0$ for all $x\in[0,x_0]$. Also, since
$P'(0)\neq0$ and $\gamma_{2,p}$ is continuous, there exists a
constant $C>0$ such that $\gamma_2(x)\leq Cx$ for all
$x\in[0,x_0]$.

Finally, we also note that if $\gamma_1$ solves (\ref{DSeqn}) and
$\gamma_2(x)$ is {\em any} positive function on $[0,x_0]$, then
for all $x\in[0,x_0]$, we have
\begin{equs}
\label{solution}
\gamma_1(x)-\gamma_2(x)&=
\myl{12}
\gamma_1(x_0)-\gamma_2(x_0)
\myr{12}~K[\gamma_1,\gamma_2](x_0,x)+
\int_{x}^{x_0}
\hspace{-2mm}
R[\gamma_2](y)~K[\gamma_1,\gamma_2](y,x)~{\rm d}y~,
\end{equs}
where
\begin{equs}
K[\gamma_1,\gamma_2](x_0,x)&\equiv
\exp\myl{18}
\int_{x_0}^{x}
\frac{1}{sz}
+\frac{P(z)}{sz\gamma_1(z)\gamma_2(z)}
{\rm d}z
\myr{18}=
\left(
\frac{x}{x_0}
\right)^{\frac{1}{s}}
\exp\myl{18}
-
\int_{x}^{x_0}
\frac{P(z)}{sz\gamma_1(z)\gamma_2(z)}
{\rm d}z
\myr{18}~.
\end{equs}
We are now in position to show that all solutions to
(\ref{DSeqn}) agree to all orders in perturbation theory with the
series solution.

\begin{theorem}
Let $\gamma_1$ and $\gamma_2$ be two solutions of (\ref{DSeqn}).
Fix $p\geq2$ and let $\gamma_{2,p}$ as above. Then there exist
constants $C$ and $C_p$ such that
\begin{equs}
\my{|}{10}\gamma_1(x)-\gamma_{2,p}(x)\my{|}{10}&\leq C_{p}
x^p~,
\label{eqn:mimicracra}\\
\my{|}{12}\gamma_1(x)-\gamma_2(x)\my{|}{12}&\leq\my{|}{12}
\gamma_1(x_0)-\gamma_2(x_0)
\my{|}{12}
\left(
\frac{x}{x_0}
\right)^{\frac{1}{s}}
\exp\myl{16}C\myl{12}\frac{1}{x_0}-\frac{1}{x}\myr{12}\myr{16}~,
\label{eqn:mamammia}
\end{equs}
where (\ref{eqn:mimicracra}) holds for all $x\in[0,x_0/2]$ and
(\ref{eqn:mamammia}) for all $x\in[0,x_0]$.
\end{theorem}

In other words, for any solution of (\ref{DSeqn}), any derivative
(of finite order) converges as $x\to0$ to the corresponding
derivative of $\gamma_{2,p}$. Hence a (truncated) power series
expansion at $x=0$ of any solution agrees to any (finite) order
with the truncated divergent series of the same order. Also, the
difference between any two solutions of (\ref{DSeqn}) decays
faster than $\ed^{-C/x}$ as $x\to0$.

\begin{proof}
In Lemma \ref{lem:ratamiaou} below, we prove that there exist
constants $c_1$ and $c_2$ such that $\gamma_1(x)\leq c_1x$ and
$\gamma_2(x)\leq c_2x$ for all $x\in[0,x_0]$ if $\gamma_1$ and
$\gamma_2$ are solutions of (\ref{DSeqn}). In particular, there
exist a constant $C>0$ such that
\begin{equs}
K[\gamma_1,\gamma_2](y,x)\leq 
\left(
\frac{x}{y}
\right)^{\frac{1}{s}}
\exp\myl{16}
-\int_{x}^{y}\frac{C}{z^2}{\rm d}z
\myr{16}
=
\left(
\frac{x}{y}
\right)^{\frac{1}{s}}
\exp\myl{16}C\myl{12}\frac{1}{y}-\frac{1}{x}\myr{12}\myr{16}
\label{eqn:hihi}
\end{equs}
for all $0\leq x\leq y\leq x_0$. Since $R[\gamma_2](x)=0$ if
$\gamma_2$ is a solution of (\ref{DSeqn}), we have
\begin{equs}
\gamma_1(x)-\gamma_2(x)=\myl{12}
\gamma_1(x)-\gamma_2(x)
\myr{12}
K[\gamma_1,\gamma_2](x_0,x)~,
\end{equs}
from which we get (\ref{eqn:mamammia}) immediately.

On the other hand, $\gamma_{2,p}$ also satisfies
$\gamma_{2,p}(x)\leq Cx$ for all $x\in[0,x_0]$, and so
$K[\gamma_1,\gamma_{2,p}]$ also satisfies
\begin{equs}
K[\gamma_1,\gamma_{2,p}](y,x)\leq 
\left(
\frac{x}{y}
\right)^{\frac{1}{s}}
\exp\myl{16}C\myl{12}\frac{1}{y}-\frac{1}{x}\myr{12}\myr{16}
\end{equs}
for all $0\leq x\leq y\leq x_0$. In particular, for fixed $x_0$,
$K[\gamma_1,\gamma_{2,p}](x_0,x)$ decays faster than any power
law as $x\to0$. Assume now that $x\leq x_0/2$, and note
that
\begin{equs}
\my{|}{10}\gamma_1(x)-\gamma_{2,p}(x)\my{|}{10}&\leq
\my{|}{10}
\gamma_1(x_0)-\gamma_{2,p}(x_0)
\my{|}{10}~K[\gamma_1,\gamma_{2,p}](x_0,x)+
C
\int_{x}^{x_0}
\hspace{-2mm}
y^{p-1}~K[\gamma_1,\gamma_{2,p}](y,x)~{\rm d}y~.
\end{equs}
Splitting the integral over $[x,x_0]$ into $[x,2x]$ and
$[2x,x_0]$, and using that $K[\gamma_1,\gamma_{2,p}](y,x)$ is a
decreasing function of $y$, we find
\begin{equs}
\int_{x}^{2x}
\hspace{-2mm}
y^{p-1}~K[\gamma_1,\gamma_{2,p}](y,x)~{\rm d}y
&\leq
\int_{x}^{2x}
\hspace{-2mm}
y^{p-1}~K[\gamma_1,\gamma_{2,p}](x,x)~{\rm d}y
=x^p\myl{10}{\textstyle\frac{2^{p}-1}{p}}\myr{10}\\
\int_{2x}^{x_0}
\hspace{-2mm}
y^{p-1}~K[\gamma_1,\gamma_{2,p}](y,x)~{\rm d}y
&\leq
K[\gamma_1,\gamma_{2,p}](2x,x)
\int_{2x}^{x_0}
\hspace{-2mm}
y^{p-1}~{\rm d}y
=K[\gamma_1,\gamma_{2,p}](2x,x)
\myl{10}
{\textstyle\frac{x_0^p-x^p}{p}}
\myr{10}~.
\end{equs}
The proof of (\ref{eqn:mimicracra}) is completed by noting that 
by (\ref{eqn:hihi}), $K[\gamma_1,\gamma_{2,p}](2x,x)$ and
$K[\gamma_1,\gamma_{2,p}](x_0,x)$ decay faster than
any power law as $x\to0$.
\end{proof}

\begin{lemma}
\label{lem:ratamiaou}
For any $\gamma_0>0$, the solution of
(\ref{DSeqn}) exists for all $x\in[0,x_0]$ and
\begin{equs}
\gamma_1(x)\leq C~x~~~~\forall~x\in[0,x_0]
\end{equs}
for some $C=C(x_0,\gamma_0)>0$.
\end{lemma}

\begin{proof}
We first note for future reference that $\frac{P(x)}{x}$ 
is continuous for all $x\in[0,x_0]$, and that there exists
constants $C_{\pm}>0$ such that $C_{-} x\leq P(x)\leq C_{+} x$ for
all $x\in[0,x_0]$. By (\ref{eqn:gnagna}), we only have to
consider $s\geq1$. We first choose $\boundC>0$
such that
\begin{equs}
\boundC\geq 
\my{\{}{18}
\begin{array}{ll}
\max(2C_{+},C_b|\ln(x_0)|) & \mbox{if }~~s=1\\[1mm]
\max\myl{12}2C_{+},\frac{1}{4x_0(s-1)},C_b^s4^{s-1}(s-1)^{s-1}\myr{12} & \mbox{if }~~s>1
\end{array}~.
\end{equs}
We then note that
\begin{equs}
f_s(\boundC~x,x)-\boundC=
\frac{1}{sx}-\frac{P(x)}{sx^2\boundC}
+\boundC\myl{10}
{\textstyle\frac{1}{s}-1}
\myr{10}\geq
\frac{1}{sx}-\frac{C_{+}}{sx\boundC}
+\boundC\myl{10}
{\textstyle\frac{1}{s}-1}\myr{10}
\geq
\frac{1}{2sx}+\boundC\myl{10}{\textstyle\frac{1}{s}-1}\myr{10}~,
\end{equs}
from which it follows that
\begin{equs}
f_s(\boundC~x,x)>\boundC~~~\forall~x
\in[0,X(s,\boundC)]~~~\mbox{where}~~~
X(s,\boundC)=
\my{\{}{18}
\begin{array}{ll}
x_0 & \mbox{if }~~s=1\\[1mm]
\frac{1}{4\boundC(s-1)} & \mbox{if }~~s>1
\end{array}~.
\label{eqn:bnubnu}
\end{equs}
We then note that
\begin{equs}[3]
\begin{array}{rcccccl}
C_b~X(1,\boundC)|\ln(X(1,\boundC))|
&=&C_b~x_0|\ln(x_0)|
&\leq& \boundC x_0&=&\boundC X(1,\boundC)~,\\
C_b~X(s,\boundC)^{1/s}&=&
\frac{C_b}{
(4\boundC(s-1))^{1/s}}
&\leq& \frac{1}{4(s-1)}
&=& \boundC X(s,\boundC)~~~\mbox{if}~~~s>1~.
\end{array}
\label{eqn:gnigni}
\end{equs}
In other words, (\ref{eqn:gnagna}) and (\ref{eqn:gnigni}) imply
that the solution of (\ref{DSeqn}) satisfies
$\gamma_1(X(s,\boundC))\leq \boundC X(s,\boundC)$. Since
$X(s,\boundC)\leq x_0$, this means that as $x$ decreases,
solutions enter the triangle
\begin{equs}
\Delta_{s,\boundC}=
\my{\{}{10}
(x,\gamma_1)~|~x\in[0,X(s,\boundC)]\mbox{ and
}0\leq\gamma_1\leq\boundC~x\my{\}}{10}
\end{equs}
through its right boundary (i.e. at $x=X(s,\boundC)$).

Solutions that enter $\Delta_{s,\boundC}$ at $x=X(s,\boundC)$
cannot satisfy $\gamma_1(x^{\star})=\boundC x^{\star}$ at some
$x^{\star}<X(s,\boundC)$, for by (\ref{eqn:bnubnu}), we would
have $\gamma_1(x)>\boundC x$ for all
$x\in(x^{\star},X(s,\boundC)]$, a contradiction with
$\gamma_1(X(s,\boundC))\leq \boundC X(s,\boundC)$. Hence
$\gamma_1(x)\leq\boundC x$ for all $x\in[0,X(s,\boundC)]$, and
the proof is completed by noting that the bound
(\ref{eqn:gnagna}) is stronger than $\gamma_1(x)\leq \boundC x$
for $x\in(X(s,\boundC),x_0]$.
\end{proof}

\end{document}